\documentclass[aos]{imsart}

\RequirePackage{amsthm,amsmath,amsfonts,amssymb,comment}
\RequirePackage[numbers,sort&compress]{natbib}
\RequirePackage[colorlinks,citecolor=blue,urlcolor=blue]{hyperref}
\RequirePackage{graphicx}

\usepackage{amsthm}
\usepackage{hyperref}
\hypersetup
{colorlinks = true,
linkcolor = red, 
anchorcolor = red, 
citecolor = blue, 
filecolor = blue,
urlcolor = blue}
\usepackage{hypcap}
\usepackage{cleveref}
\usepackage{lipsum}
\usepackage{enumerate}
\usepackage{enumitem}
\usepackage{amsmath,amsfonts,amssymb,mathrsfs, amscd,amsthm,amsbsy,amsxtra,bbm,bm, epsf,calc,comment, xcolor}
\usepackage[toc,page]{appendix}
\usepackage{color}
\usepackage{datetime}
\usepackage{latexsym}
\usepackage[english]{babel}
\usepackage{graphicx}
\usepackage{epsfig}
\usepackage{dsfont}
\usepackage{tikz}

\startlocaldefs
\theoremstyle{plain}

\newtheorem{theorem}{Theorem}[section]
\newtheorem{lemma}[theorem]{Lemma}
\theoremstyle{remark}
\newtheorem{definition}[theorem]{Definition}
\newtheorem*{example}{Example}

\endlocaldefs

\definecolor{mygreen}{rgb}{0.1,0.75,0.2}

\DeclareSymbolFont{bbold}{U}{bbold}{m}{n}
\DeclareSymbolFontAlphabet{\mathbbold}{bbold}

\newcommand{\spt}{\textup{spt}}

\newcommand{\R}{\mathbb{R}}

\numberwithin{equation}{section}
\newtheorem{proposition}[theorem]{Proposition}

\theoremstyle{remark}
\newtheorem{remark}[theorem]{Remark}

\begin{document}

\begin{frontmatter}
\title{Statistical inference of convex order by Wasserstein projection}
\runtitle{Statistical inference of convex order}

\begin{aug}
\author[A]{\fnms{Jakwang}~\snm{Kim}\ead[label=e1]{jakwang.kim@math.ubc.ca}},
\author[A]{\fnms{Young-Heon}~\snm{Kim}\ead[label=e2]{yhkim@math.ubc.ca}}
\author[C]{\fnms{Yuanlong}~\snm{Ruan}\ead[label=e3]{ruanyl@buaa.edu.cn}}
\and
\author[A]{\fnms{Andrew}~\snm{Warren}\ead[label=e4]{awarren@math.ubc.ca}}
\address[A]{Department of Mathematics, University of British Columbia, Vancouver, British Columbia, Canada \printead[presep={,\ }]{e1,e2,e4}}


\address[C]{School of Mathematical Sciences, Beihang University, Beijing, China \printead[presep={,\ }]{e3}}
\end{aug}

\begin{abstract}
Ranking distributions according to a stochastic order has wide applications in diverse areas. Although stochastic dominance has received much attention, convex order, particularly in general dimensions, has yet to be investigated from a statistical point of view. This article addresses this gap by introducing a simple statistical test for convex order based on the Wasserstein projection distance. This projection distance not only encodes whether two distributions are indeed in convex order, but also quantifies the deviation from the desired convex order and produces an optimal convex order approximation. Lipschitz stability of the backward and forward Wasserstein projection distance is proved, which leads to elegant consistency and concentration results of the estimator we employ as our test statistic. Combining these with state of the art results regarding the convergence rate of empirical distributions, we also derive upper bounds for the $p$-value and type I error of our test statistic, as well as upper bounds on the type II error for an appropriate class of strict alternatives. With proper choices of families of distributions, we further attain that the power of the proposed test increases to one as the number of samples grows to infinity. Lastly, we provide an efficient numerical scheme for our test statistic, by way of an entropic Frank-Wolfe algorithm. Experiments based on synthetic data sets illuminate the success of our approach.

\end{abstract}



\end{frontmatter}

\section{Introduction}
Decision makers constantly face the question of making optimal choices under uncertainty. As a fundamental step towards such decision-making, determining certain stochastic order or ranking of distributions in question arises in a number of scenarios. Typically, a stochastic order is a reflection of preference, for example in economics. In other scenarios, stochastic ordering is not part of a decision-making procedure, but instead serves as a basic assumption under which subsequent investigations are performed. We refer the readers to \cite{shaked2007stochastic, stochasticorder2018} and references therein for more details. In either situation, it is essential to have statistical tools to find out the order of distributions or test whether they have the anticipated order.

Precisely, a stochastic order between two distributions $\mu$ and $\nu$ on
$\mathbb{R}^{d}$ is defined through their actions on a given set $\mathcal{A}%
$\ of admissible test functions on $\mathbb{R}^{d}$,%
\begin{equation}
\int\varphi   d\mu\leq\int\varphi 
d\nu\text{ for all }\varphi\in\mathcal{A}. \label{Inq:StochDomClass}
\end{equation}
Then $\mu$ is said to dominate $\nu$ in the stochastic order determined by
$\mathcal{A}$, written $\mu\preceq_{\mathcal{A}}\nu$.

In economics, finance, actuarial science and other fields relating to the theory of risk, an admissible set $\mathcal{A}$ consists of test functions that reflect certain preference, which are termed \emph{utilities}. Consider, for example, on $\mathbb{R}$, 
\[
\mathcal{A}_{k}\mathcal{=}\left\{  \varphi: \mathbb{R\mapsto R}:\left(
-1\right)  ^{i+1}\varphi^{\left(  i\right)  }\geq0,i=1,...,k\right\}
,\text{ }k\geq1.
\]
Each $\mathcal{A}_{k}$ produces a stochastic order, usually called the $k$-th
order \emph{stochastic dominance}.
The first order stochastic dominance which compares mean values of two distributions is frequently used by economists to study welfare, poverty and inequality \cite{davidson2000statistical}. In the context of investment, $\mu$ and $\nu$ could be distributions of returns from two investment strategies or portfolios. Then the first order stochastic dominance\ gives investors partial information to decide which strategy to choose depending on their own preferences for return. Higher order stochastic dominance serves similar purposes of providing partial information to investors. The preferences of investors are encoded in the utility class $\mathcal{A}_{k}$. The application of stochastic order in the context of economics and finance can be found in several papers: see \cite{whitmore1970third, whitmore1978stochastic, menezes1980increasing, ekern1980increasing, shorrocks1987transfer}.

Owing to its importance, the study of stochastic dominance including statistical testing of any pre-specified order has received much attention and is now studied well. We refer to \cite{mcfadden1989testing,  barrett2003consistent, linton2005consistent} 
 and references
therein for discussions in this respect. Comparatively, \emph{convex order}, which is still widely used in the literature (see \cite{shaked2007stochastic, levy2015stochastic}), has not been well-studied from a statistical standpoint, and we will focus on it in this paper. A distribution $\mu$ on $\mathbb{R}^{d}$ is said to dominate another distribution $\nu$ on $\mathbb{R}^{d}$ in convex order, written $\mu\preceq\nu,$ if
\eqref{Inq:StochDomClass} holds for $\mathcal{A}$ being the set of all convex functions on $\mathbb{R}^{d}$, i.e.,
\begin{equation}\label{Inq:CvxClass}
    \int\varphi  d\mu\leq\int\varphi d\nu\text{ for all }\varphi\text{ convex on }\mathbb{R}^{d}.
\end{equation}
Note that convex order is related to, but different from the above $k$-th order stochastic
dominance. However, when comparing distributions of equal mean, convex order will serve the same purpose as the second order stochastic dominance, e.g., in selecting portfolio according to risk preference. Our main question in this paper is the following:
\begin{align*}
    &\hbox{How do we statistically and quantitatively test the convex order relation}\\ 
   & \hbox{between two distributions, in an effective way?}
\end{align*}

To give some motivations for this question, we first note that convex order has a natural connection with martingales. A necessary and
sufficient condition for two distributions to be in convex order is that there exists a martingale coupling between them \cite{strassen1965existence}. On top
of this fact, several research directions have emerged. Among them are model-independent bounds for derivative pricing \cite{beiglbock2013model, galichon2014stochastic}, model-free version of the fundamental theorem of asset pricing \cite{acciaio2016model} and martingale optimal transport: see \cite{hobson2012robust, MR3127880, MR3161649, dolinsky2014martingale, Beiglb_Juillet16, henry2016explicit, campi2017change, completedualityMOT17, hou2018robust, structureMOT19,  IrreducibleMOT19,  computationalMOT19, MR4152642, cheridito2021martingale}.

In all these studies, the convex order is unanimously used as a basic assumption for further queries, however, the study on the convex order itself is rare. Clearly, this raises the question of determining whether such an assumption is indeed satisfied. This question becomes even more prominent when attempting to apply these upstream theories to practical problems. Therefore, be it in theory or application, it is an important initial step to test whether two distributions are in convex order. It is one
of the purposes of the article to fill this gap, i.e. to design a simple yet powerful statistical tool for convex order in general dimensions.

Another purpose of designing a statistical tool for convex order is driven by applications in decision making. One such example is testing of arbitrage. The Fundamental Theorem of Asset Pricing shows that the absence of arbitrage is equivalent to the existence of a risk-neutral martingale probability (that is equivalent to the physical market measure). The risk-neutral probability is unique if the market is complete. In practice, a risk-neutral martingale probability is not directly observable but instead marginally obtained via calibrations to option prices trading on the market. Specifically, under the market completeness assumption, given option prices with a continuum of strike prices at two maturities $T_{1}<T_{2},$ the classical procedure \cite{breeden1978prices, banz1978prices} then extracts the marginal distributions, say $\mu_{T_{1}}$ and $\mu_{T_{2}},$ of the risk-neutral martingale probability implied by the market. Note that the procedure
only gives us marginal distributions, not the whole risk-neutral martingale probability. In reality, the procedure can be hampered by practical reasons, e.g., the liquidity of the option contracts and that option prices are only available for a short range of strike prices \cite{figlewski2008estimating, figlewski2018risk}. Thus, a good extracting procedure should address these difficulties and importantly be consistent across maturities, i.e., $\mu_{T_{1}}\preceq\mu_{T_{2}}.$ Clearly being able to statistically test the convex order relationship between $\mu_{T_{1}},$ $\mu_{T_{2}}$ would be useful for the task of tuning an effective extracting procedure. Additionally, a detectable deviation from the expected convex order that $\mu_{T_{1}}\preceq\mu_{T_{2}}$ would imply the existence of arbitrage. Although sometimes arbitrage opportunities can be beneficial for financial markets, they can also impair market liquidity \cite{foucault2017toxic}. Detection of arbitrage, equivalently, violation of the convex order between the marginals, will enable policy makers to enforce regulatory constraints so as to contain harmful arbitrage opportunities.

Convex order has equally intriguing development in economics. A classical question in labor economics is the production relationship between workers and firms. The central question is how workers are matched to firms in such a way that a country's production output is maximized. A worker may possess different or multidimensional skills, while a firm needs a combination of selected skills to complete a task. A firm may be interested in one of the various skills of a worker, but it can only hire the worker as a whole, together with other less interested skills. The fact that a worker's skills cannot be \emph{physically unbundled} has made the labor economics problem difficult to solve.  An elegant idea, which is recently introduced in \cite{chone2021matching, nordstrom2022workers} is to
\emph{mathematically unbundle} a worker's skills: the ideal (unrealistic) situation in which a firm has access to all skill components a person is endowed with. Of course, it is impossible for firms to access each individual skill of a person. Instead, the authors show that firms choose the closest wage distribution among those which are dominated by the distribution of workers' aggregated skills in convex order to minimize the production cost, and a pair of that closest wage distribution and workers' skills forms a market equilibrium. In this context, testing of convex order between wage distribution and distribution of aggregated skills can be leveraged to test whether a labor market is at an equilibrium, and would clearly be valuable to governments when policies are being made.


With these domain applications in mind, we intend to design a statistical tool for convex order. A
good convex order test statistic should be naturally consistent, efficient, and easy to compute. Moreover, it should provide more information than just telling us to accept or not the hypothesis. Expanding our goals, we will address the following.

\begin{enumerate}
\item[$\left(  G_{1}\right)  $] Given unknown probability measures $\mu$, $\nu \in
\mathcal{P}_2\left(  \mathbb{R}^{d}\right)$, consider the \emph{hypothesis testing of convex order},
\begin{equation}\label{def : convex order test}
     \boldsymbol{H_0} :  \mu \preceq \nu \quad \text{vs} \quad \boldsymbol{H_A} : \mu  \npreceq \nu
\end{equation}
by constructing a test statistic depending on $X_1, \dots, X_n \sim \mu$ and $Y_1, \dots, Y_m \sim \nu$ which are drawn identically and independently (but, no particular dependence relation between $X_i$'s and $Y_j$'s is assumed).

\item[$\left(  G_{2}\right)  $] If they are not in convex order, how much do
they deviate from the expected convex order relationship? This feature will be
useful for decision-maker when decisions are made under pre-specified tolerance.

\item[$\left(  G_{3}\right)  $] If they are not in convex order and one of
them is allowed to change to form a convex order relationship, what is the
closest modification? This is very useful in numerical computation involving
measures in convex order, providing an indirect way to sample measures in
convex order.
\end{enumerate}
We emphasize that the unknown $\mu$ and $\nu$ which generate $X_i$'s and $Y_j$'s respectively are not assumed to be coupled by a certain martingale. Rather, the goal of \eqref{def : convex order test} is to test the \emph{possibility} of martingale coupling between the two.

At this point, it is worthwhile distinguishing the hypothesis testing $\left(
\ref{def : convex order test}\right)  $ from the following test recently studied in
\cite{blanchet2024empirical}. Let $\left(  X,Y\right)  \sim\mathbb{P}$ with $\mathbb{P}$ having marginals $\mu$ and $\nu$. Consider the \emph{hypothesis testing of martingale pair},
\begin{equation}\label{eq:martingale_pair_test}
    \boldsymbol{H_0'}:\left(  X,Y\right)  \text{ is a marginale pair under }
\mathbb{P}\text{.}
\end{equation}

Testing of convex order $\boldsymbol{H_0}$ is related to but different from the martingale test $\boldsymbol{H_0'}$. In view of the relationship between convex order and existence of martingale coupling, if one rejects $\boldsymbol{H_0}$, then one can also reject the martingale test $\boldsymbol{H_0'}$ because rejecting $\boldsymbol{H_0}$ implies the impossibility of a martingale coupling. On the other hand, as a contraposition of the previous statement, if one cannot reject $\boldsymbol{H_0'}$, so as $\boldsymbol{H_0}$. While even if the martingale test $\boldsymbol{H_0'}$ fails, it does not necessarily exclude $\boldsymbol{H_0}$ either. Convex order testing is useful in situations where only marginal distributions are available while the joint distribution is not easy to obtain, both the above-mentioned examples of arbitrage detection and labor economics well fit into this situation. If the joint distribution is obtainable, for example modeled by external methods \cite{gierjatowicz2023robust}, then the hypothesis testing \eqref{eq:martingale_pair_test} of martingale pair would be easier.

A first thought about the testing of convex order may prompt one to propose a test directly based on the definition \eqref{Inq:CvxClass},
however this simply turns out to not work, let alone satisfy the goals $\left(  G_{1}\right)  ,\left(  G_{2}\right)$ and $\left(  G_{3}\right).$ This is demonstrated by the example below.

\begin{example}\label{ex: example}
For two probability measures, $\mu$ and $\nu$,
if $\mu$ and $\nu$ are not in convex order, by the definition of convex order it holds that
\[  
    \sup_{\varphi \in \mathcal{A}} \left\{ \int_{\mathbb{R}^d} \varphi d\mu - \int_{\mathbb{R}^d} \varphi d\nu \right\} = \infty.
\]
Let $\mu_n$ and $\nu_n$ be empirical measures of $\mu$ and $\nu$, respectively.
The definition seemingly provides a naive test statistic for convex order, \eqref{def : convex order test}, which is
\[
    T_n: = \sup_{\varphi\in\mathcal{A}}\left\{  \int\varphi d\mu_n -\int\varphi d\nu_n\right\}.
\]
In contrary to its natural derivation, it is not robust and does not work as expected.

Consider probability measures in two dimension, $\mu=\delta_{\left(  0,0\right)  }$ and $\nu=\frac{1}{2}\delta_{\left(  -1,0\right)} + \frac{1}{2}\delta_{\left(  1,0\right)  }$. It is easy to check that $\mu \preceq \nu.$ Now, we approximate $\nu$ by the sequence
\[
    \nu_n:=\frac{1}{2}\delta_{\left(  -1,0\right)  }+\frac{1}{2}\delta_{\left(1,\frac{1}{n}\right)  } \longrightarrow \nu \text{ weakly as $n \to \infty$.}%
\]
Clearly $\mu  \npreceq \nu_n$ for all $n \geq 1$. In fact, taking the set of convex functions which are zero on the line connecting $\left(-1,0\right)$ and $\left(  1,\frac{1}{n}\right)$, and arbitrarily large at $\left(0,0\right)$ yields $T_n = \infty$. Therefore, $T_n \not \to 0$ for any $n$. This shows that $T_n$ does not provide the expected information for determining the convex order relation between $\mu$ and $\nu.$ 
\end{example}

To achieve our goals, we need a statistic which can not only perform the test, but also quantify the deviation from an expected test result. The key components in achieving our goals are the backward and forward Wasserstein projections:
\[
    \mathcal{W}_{2}\left(\mu,\mathscr{P}^{cx}_{\preceq \nu}\right):=\inf_{\xi \in \mathscr{P}^{cx}_{\preceq \nu}}\mathcal{W}_{2}\left(\mu,\xi\right),\quad \mathcal{W}_{2}\left(\mathscr{P}_{\mu \preceq}^{cx},\nu\right) :=\inf_{\eta \in \mathscr{P}_{\mu \preceq}^{cx}}\mathcal{W}_{2}\left(\eta,\nu\right).
\]
Here $\mathscr{P}^{cx}_{\preceq \nu}$ and $\mathscr{P}_{\mu \preceq}^{cx}$ are backward and forward convex order cones defined as
\[
    \mathscr{P}^{cx}_{\preceq \nu}:=\left\{  \xi\in \mathcal{P}_2(\mathbb{R}^{d}) : \xi\preceq\nu\right\},\quad \mathscr{P}_{\mu \preceq}^{cx}:=\left\{  \eta\in \mathcal{P}_2(\mathbb{R}^{d})  :\mu \preceq\eta \right\}  .
\]
Wasserstein projections for $\mathcal{W}_p$, $p \geq 1$ are introduced and studied in \cite{Alfonsi_2020}, and further investigated in \cite{Gozlan_Juillet20}. The dualities for the backward and forward Wasserstein projections are established in \cite{yh_yl_stochastic_order}. It is also shown that backward and forward Wasserstein projections are inherently connected \cite{Alfonsi_2020, yh_yl_stochastic_order}. These properties enable us to design a test statistic for convex order which we present below.

To the best of our knowledge, no statistical tool for convex order as required by $\left(  G_{1}\right),\left(  G_{2}\right)$ and $\left(  G_{3}\right)$ has been developed in the literature. A major difficulty of the achievement of such desirable statistic comes from the convex order itself. In particular, stability results are required to obtain the consistency of test statistic. This involves more or less the stability of the (backward and forward) convex order cones, $\mathscr{P}^{cx}_{\preceq \nu}$ and $\mathscr{P}_{\mu \preceq}^{cx}$, w.r.t. marginals $\mu$ and $\nu$ in some sense. This, however, turns out to be not obvious. The sense in which the convex order cones are stable is not clear as there are examples showing that $\mathscr{P}^{cx}_{\preceq \nu_m}$ could be very different from $\mathscr{P}^{cx}_{\preceq \nu}$ even if $\nu_m$ converges weakly to $\nu$ \cite{bruckerhoff2022instability}. It should be noted that some of the stability issues have been appreciated in the study of martingale pair testing \cite{blanchet2024empirical}, where the authors propose a vanilla version of the martingale pair test of \eqref{eq:martingale_pair_test} but find that it is not consistent, therefore we turn to study a parametrically smoothed statistic. The choice of smoothing kernel parameters may affect the statistical test result.

\section{Main results}\label{sec: main results}


Our first main result, presented in Theorem~\ref{thm: consistency of projected Wasserstein-1}, is the quantitative stability of the backward and forward Wasserstein projection
distances (note calling them distances need justification), i.e. how $\mathcal{W}_{2}\left(  \mu,\mathscr{P}^{cx}_{\preceq \nu}\right)$ and $\mathcal{W}_{2}\left(\mathscr{P}_{\mu \preceq}^{cx},\nu\right)$ change w.r.t. the perturbations of marginals. It is shown that they are not only stable w.r.t. perturbation of marginals, but also has stronger stability. In fact, they are Lipschitz w.r.t. the martingals. It is worth mentioning that there is a connection between the \emph{backward} Wasserstein projection $\mathcal{W}_{2}\left(  \mu,\mathscr{P}^{cx}_{\preceq \nu}\right)$ and martingale optimal transport
\cite{Gozlan_Juillet20}. A qualitative stability of the weak optimal transport is given in \cite{backhoff2022stability}. But in order to design a good test statistic, qualitative stability is not enough, particularly in terms of obtaining subsequent statistical bounds. In view of the above connection, our results on backward Wasserstein projection thus provide the first quantitative stability for the weak optimal transport with barycentric cost, and it may also have
applications in problems related to weak optimal transport: see \cite{backhoff2022applications}.

Theorem~\ref{thm: consistency of projected Wasserstein-1} gives a hint for a good candidate statistic for \eqref{def : convex order test}. Given unknown $\mu$ and $\nu$, we consider their empirical measures $\mu_n$ and $\nu_m$, and propose the empirical Wasserstein projection distance $\mathcal{W}_2(\mu_n, \mathscr{P}^{cx}_{\preceq \nu_m})$ as a test statistic. The proposed test for \eqref{def : convex order test} is
\begin{align}\label{eqn:test-new-t}
    \text{Reject $\boldsymbol{H_0}$ if $\mathcal{W}_2(\mu_n, \mathscr{P}^{cx}_{\preceq \nu_m}) \geq t(\alpha)$}; \text{ Accept otherwise},
\end{align}
where $\alpha \in (0,1)$ is the significance level. The quantity $t(\alpha)$ is chosen so that the type I error, the probability of false rejection, cannot be larger than $\alpha$. From the stability of $\mathcal{W}_2(\mu, \mathscr{P}^{cx}_{\preceq \nu})$, it is deduced that 
$\mathcal{W}_{2}\left(  \mu_n,\mathscr{P}_{\preceq \nu_m}^{cx}\right)$ converges to the true one as $n,m \to \infty$, hence it is a consistent estimator of  $\mathcal{W}_2(\mu, \mathscr{P}^{cx}_{\preceq \nu})$. Furthermore, we derive its convergence rate under mild assumptions. There are two possible scenarios: log-Sobolev inequality and bounded supports. Notice that the log-Sobolev case has the rate $O\left( n^{-1/d} \right)$ while the bounded supports case has $O\left( n^{-1/k} \right)$ where $k$ is in some sense the intrinsic dimension of the supports of probability measures. So, if one knows that the support of a measure is bounded and has a lower intrinsic dimension, the sample complexity is lower than the usual order. On the other hand, there are many measures which do not have bounded support, e.g. Gaussian convoluted measures. Then, the log-Sobolev case covers a class of such measures, albeit paying the larger rate $O\left( n^{-1/d} \right)$. 



%
%
%
%
%


To obtain the rate of convergence, we study the tail probability of the discrepancy between $\mathcal{W}_{2}\left(  \mu_n,\mathscr{P}_{\preceq \nu_m}^{cx}\right)$ and $\mathcal{W}_{2}\left(  \mu,\mathscr{P}_{\preceq \nu}^{cx}\right)$ in Proposition \ref{prop: tail_probability}, which relies on a new concentration
result for the Wasserstein projection distances in Lemma~\ref{lem: concentration}. From this, the rate of convergence of the test statistic will be given in Theorem \ref{thm : Convergence rate}.

It provides guidance on how to control type I/II errors, therefore a way to attain the consistency of the proposed test \eqref{eqn:test-new-t}. In Section \ref{sec: p-value, errors} we will obtain estimates on the $p$-value and type I/II errors of $\mathcal{W}_{2}\left(  \mu_n,\mathscr{P}_{\preceq \nu_m}^{cx}\right)$  under mild restrictions on the distributions $\mu$ and $\nu$ in our statistical model: see Propositions \ref{prop: p-value, Type I error} and \ref{prop: type II error}. These will give Theorem~\ref{thm: consistency of the test} regarding the consistency of the proposed test, controlling type I and II errors simultaneously, and achieving vanishing error rates.

Our last contribution is about the computation of $\mathcal{W}_{2}\left(  \mu_n,\mathscr{P}_{\preceq \nu_m}^{cx}\right)$. We first convert the problem into a quadratic constrained optimization as in \cite{Alfonsi_2020, Gozlan_Juillet20}. To develop an efficient optimization scheme, we introduce an iterative method \emph{entropic Frank-Wolfe algorithm} in Section~\ref{sec:experiment}. It combines the popular Frank-Wolfe algorithm with entropic regularization. Note that it is different from the usual entropic optimal transport in the computational optimal transport community: our method does not regularize the objective function but does during the optimization procedure. We show its efficiency empirically on the basis of a synthetic data set.

Let us briefly compare some of our results with \cite{Alfonsi_2020}. Both \cite{Alfonsi_2020} and our paper have addressed problems related to the Wasserstein projection distance. But the purposes are different. In \cite{Alfonsi_2020}, the main goal is to sample an optimal martingale coupling to solve maximization problems with the martingale constraint, which typically aim at financial applications. The authors show the stability of the Wasserstein projection distances \emph{between empirical distributions} drawn from two probability measures \emph{which are already in convex order, i.e., $\mu \preceq \nu$}: see \cite[Propositions 3.1 and 4.3]{Alfonsi_2020}.

In the present paper, we prove in Theorem~\ref{thm: consistency of projected Wasserstein-1}, the stability of the Wasserstein projection distance for \emph{any probability measures}, namely without the convex order assumption, thus generalize the results of \cite{Alfonsi_2020} mentioned in the above paragraph. Our stability result is necessary for the application to convex order hypothesis problem. It enables us to obtain desirable properties of the proposed statistic, such as sharp consistency, type I error estimate, and balanced type II error control. A new concentration for the Wasserstein projection distance is obtained in the course.

Regarding the Wasserstein projection measure, not just the corresponding projection distance, it is challenging to show its stability in general dimensions; in the one dimensional case such a stability is obtained in \cite{Jourdain_2023}. Since its proof is tailored to one dimension, it seems not trivial to extend their result to general dimensions.

\subsection{Organization}
In Section \ref{sec:preplim}, some background of optimal transport and our main contributions will be provided. We also review a key technique, Wasserstein projection, and related important lemmas. Then, we will prove the Lipschitz stability of $\mathcal{W}_2(\mu, \mathscr{P}^{cx}_{\preceq \nu})$ on Wasserstein space in Section \ref{sec:stability}. Combining the stability result with the state of the art results of convergence rate of the Wasserstein distance of empirical distributions, we will provide the rate of convergence of $\mathcal{W}_2(\mu_n, \mathscr{P}^{cx}_{\preceq \nu_m})$ in Section \ref{sec : Convergence rate}. upper bounds of the $p$-value and type I/II errors, and the consistency of the proposed test in Section \ref{sec: p-value, errors}. In Sections \ref{sec:computation} and \ref{sec:experiment}, we will give an exact formulation for computing $\mathcal{W}_2(\mu_n, \mathscr{P}^{cx}_{\preceq \nu_m})$, an efficient algorithm and some numerical evidences based on the formulation. Lastly, we will finish this paper and mention some future direction and related interesting questions in Section \ref{sec: conclusion}. Proofs that are not provided in the main body of the paper are put in the appendix or supplements, these include \Cref{lem: concentration}, \Cref{prop: tail_probability}, Lemma \ref{lemma : barycentric characterization} and Theorem \ref{thm : computing projection of mu_n}.


\subsection{Notation}
We define in the following some of the notation that will be used in the present paper.

\begin{itemize}
    \item $\mathcal{P}_p\left(  \mathbb{R}^{d}\right)$ is the set of probability measures over $\mathbb{R}^d$ with
finite $p$-th order moment.
    \item $\Pi(\mu, \nu)$ is the set of couplings, i.e. the set of joint distributions whose first and second marginals are $\mu$ and $\nu$, respectively.
    \item $\mu_n$, $\nu_m$ are empirical measures of $\mu, \nu$, respectively,  made by samples $X_1, \cdots, X_n \sim \mu$, $Y_1, \cdots, Y_m \sim \nu$,  drawn identically and independently (and, no particular dependence relation between  the sames $X_i$'s and $Y_j$'s is assumed).
    \item $(T)_{\#}(\mu)$ is a pushforward measure by a measurable map $T$, i.e. $(T)_{\#}(\mu)(B) := \mu \left( T^{-1}(B) \right)$ for any measurable subset $B$.
    \item $\mathcal{W}_{2}\left(  \mu,\nu\right)$ is the 2-Wasserstein distance between $\mu$ and $\nu$.
    \item $\mathcal{A}$ is the set of all convex functions.
    \item $\mathscr{P}^{cx}_{\preceq \nu}$ is the set of probability measures
dominated by $\nu$ in convex order.
    \item $\mathscr{P}_{\mu \preceq}^{cx}$ is the set of probability measures which
dominate $\mu$ in convex order.
    \item $x \wedge y:= \min\{x,y\}$ and $x \vee y:= \max\{x,y\}$.
\end{itemize}

\section{Preliminaries}\label{sec:preplim}
In this section, we give a quick review of the backward and forward Wasserstein projection, the weak optimal transport and a few facts related to our statistical inference.

\subsection{Wasserstein convex order projection}\label{subsec: Wasserstein projection}
In \cite{Gozlan_Juillet20}, a characterization of martingale optimal transport in terms of the backward projection is studied. In \cite{yh_yl_stochastic_order}, duality theorems are developed for Wasserstein projections in general stochastic order, which includes
(increasing) convex order, concave order and many other stochastic orders usually seen in practice. The approach taken there is closely related to
probabilistic potential theory where a probability measure which dominates another probability measure with respect to some function class is called \emph{balayage}. In this section we give a quick review of the results that will be used later on. Recall that $\mathcal{A}$ is the set of all convex functions on $\mathbb{R}^d$.

\begin{theorem}
\label{thm:duality}\cite[Theorems 5.3, 6.1 and 6.2]{yh_yl_stochastic_order}
Let $\mu, \nu \in \mathcal{P}_2 \left( \mathbb{R}^d \right)$.
\begin{enumerate}
    \item[(i)] The duality for backward projection%
\[
    \mathcal{W}^2_{2}\left(  \mu,\mathscr{P}^{cx}_{\preceq \nu}\right) =\mathcal{D}_{2}\left(  \mu,\mathscr{P}^{cx}_{\preceq \nu}\right) :=\sup_{\varphi\in\mathcal{A}}\left\{  \int_{\mathbb{R}^d} Q_{2}\left(  \varphi\right)  d\mu-\int_{\mathbb{R}^d}\varphi d\nu\right\}
\]
    where $Q_{2}\left(  \varphi\right)  \left(  x\right)  =\inf_{y\in \mathbb{R}^d}\left\{\varphi\left(  y\right)  +\left\vert x-y\right\vert ^{2}\right\}$. Furthermore, the supremum is achieved by $\varphi_0$ which is $\nu$-integrable.
    \item[(ii)] The duality for forward projection%
\[
    \mathcal{W}^2_{2}\left(  \mathscr{P}_{\mu \preceq}^{cx},\nu\right) =\mathcal{D}_{2}\left(  \mathscr{P}_{\mu\preceq}^{cx},\nu\right):=\sup_{\varphi\in\mathcal{A}}\left\{  \int_{\mathbb{R}^d}\varphi d\mu-\int_{\mathbb{R}^d}Q_{\bar{2}}\left( \varphi\right)  d\nu\right\}
\]
where $Q_{\bar{2}}\left(  \varphi\right)  \left(  y\right)  =\sup_{x\in \mathbb{R}^d}\left\{
\varphi\left(  x\right)  -\left\vert x-y\right\vert ^{2}\right\}$. Similarly, the supremum is achieved by $\varphi_0$ which exhibits at most quadratic growth.
\end{enumerate}
\end{theorem}

\begin{remark}
The convex order cones $\mathscr{P}^{cx}_{\preceq \nu}$, $\mathscr{P}^{cx}_{\mu \preceq }$ and the function class $\mathcal{A}$ clearly depend on the underlying space, but for notational simplicity, we
do not explicitly specify it unless it is ambiguous to do so. 
\end{remark}

\begin{remark}\label{rmk:thm_duality}
A nice feature of the above duality is that the
specific testing convex function class $\mathcal{A}$ itself is
not essential. It could be replaced with other convex function class
$\mathcal{A}$ that defines the same convex order as $\mathcal{A}$. See \cite[Remark 4.9]{yh_yl_stochastic_order} for details.
\end{remark}


Next we consider a type of Brenier's theorem \cite{Brenier91} and Caffarelli’s contraction theorem \cite{Caffarelli_contraction} adapted to our context. We state \cite[Theorem 1.2]{Gozlan_Juillet20} and \cite[Theorems 7.4 and 7.6]{yh_yl_stochastic_order} as one lemma, in which the existence of transport maps and their contraction properties arising in martingale optimal transport problems are studied in the language of the Wasserstein projections onto cones.

\begin{lemma}\cite[Theorem 1.2]{Gozlan_Juillet20}\cite[Theorems 7.4 and 7.6]{yh_yl_stochastic_order}\label{thm : characterization of projected measure}
For any $\mu, \nu \in \mathcal{P}_2(\mathbb{R}^d)$, there exists a unique backward projection of $\mu$, denoted by $\overline{\mu}$, onto $\mathscr{P}_{\preceq \nu}^{cx}$. Furthermore, 
\begin{equation*}\label{eq: characterization of projection measure}
    \overline{\mu} = (\nabla \varphi)_{ \#} \mu,
\end{equation*}
where $\varphi$ is a proper lower semicontinuous convex function such that $D^2 \varphi \leq  Id$.

Similarly, there is a forward projection of $\nu$ onto $\mathscr{P}_{\mu \preceq}^{cx}$. If we further assume that $\nu$ is absolutely continuous, then the projection is unique, denoted by $\overline{\nu}$. Moreover
\[
    \overline{\nu} = (\nabla \phi)_{ \#} \nu
\]
by some proper lower semicontinuous convex function $\phi$ such that $D^2 \phi \geq  Id$
\end{lemma}

\begin{remark}
We emphasize that the characterization and the uniqueness of backward projected measure $\overline{\mu}$ does not require the absolute continuity of $\mu$; while to discuss a unique forward projection $\overline{\nu}$, it requires the absolute continuity of $\nu$. This is because, intuitively, the singularity of $\mu$ can be preserved via the backward projection map $\mu \mapsto \overline{\mu}$. But the forward projection map $\nu \mapsto \overline{\nu}$ needs $\overline{\nu}$ less singular than $\nu$: it hinders the existence of a transport map. In fact, the existence of a transport map in the backward projection is crucially used in Theorem~\ref{thm : computing projection of mu_n}.
\end{remark}

The next theorem is \emph{backward-forward swap} property which is first proved by using primal formulation in \cite[Corollary 4.4]{Alfonsi_2020}, and also obtained by duality in \cite[Theorem 8.3]{yh_yl_stochastic_order}. The
property states that the backward and forward convex order projection costs are equal. This is not trivial because the backward and forward convex order cones possess distinct geometric properties: the backward cone is geodesically convex though the forward one is not. We refer \cite[Section 8.1]{yh_yl_stochastic_order} to readers in which examples are provided. In particular, it is one of the key properties for this work especially in Section~\ref{sec:stability}.

\begin{theorem}{\cite[Corollary 4.4]{Alfonsi_2020}\cite[Theorem 8.3]{yh_yl_stochastic_order}}\label{thm:swap_prop}
For any $\mu, \nu \in \mathcal{P}_2(\mathbb{R}^d)$, it holds that
\begin{equation}\label{eq : forward and backward are equal}
   \mathcal{W}_{2}\left(  \mu,\mathscr{P}_{\preceq \nu}^{cx}\right) 
   = \mathcal{W}_{2}\left( \mathscr{P}_{\mu \preceq}^{cx}, \nu\right). 
\end{equation}  
\end{theorem}

\begin{remark}\label{rem:Wp-swap}
Under suitable integrability conditions, the backward-forward swap property holds for a large class of cost functions including $|\cdot|^p$ for $1\leq p < \infty$. With appropriate modifications, all the results below can be extended to $\mathcal{W}_p$ for $1\leq p < \infty$: see \cite{Alfonsi_2020,yh_yl_stochastic_order} for more details.
\end{remark}

\subsection{Weak optimal transport and martingale constraint}\label{subsec: weakOT}
One variant of optimal transport problem which emerges from mathematics finance \cite{hobson2012robust, beiglbock2013model, galichon2014stochastic, dolinsky2014martingale, MR3127880, MR3161649, acciaio2016model, Beiglb_Juillet16, henry2016explicit, campi2017change, hou2018robust,  MR4152642, cheridito2021martingale} is \emph{martingale optimal transport problem}, which is probably the most important case of \emph{weak optimal transport problem} introduced in \cite{GOZLAN20173327, Alibert_2019}. They aim at developing robust finance, or finding a model-free upper (resp. lower) bound for the price of a given option which maximizes (resp. minimizes) the expected profit (resp. cost). Many aspects of martingale optimal transport are studied, regarding basic structural properties, duality, and computations for various necessities from different fields, especially finance and economics~\cite{MR3639595, completedualityMOT17, structureMOT19, IrreducibleMOT19, computationalMOT19, MR4171389, Gozlan_Juillet20, bruckerhoff2022instability, backhoff2022stability, yh_yl_stochastic_order, choné2024weakoptimaltransportunnormalized}.

Another instance of weak optimal transport is induced by the barycentric cost. In this case the weak optimal transport takes the form
\[
    T_2(\nu|\mu) := \inf_{\pi \in \Pi(\mu, \nu)} \int \left( \int y p_x(dy) - x \right)^2 \mu(dx).
\]
Notice that $T_2(\nu|\mu)=0$ if and only if $\mu \preceq \nu$ if and only if there is a martingale coupling between $\mu$ and $\nu$. For any $\mu, \nu \in \mathcal{P}_2(\mathbb{R}^d)$, it is proved in \cite[Proposition 1]{Gozlan_Juillet20} that
\[
    \mathcal{T}_2(\nu|\mu) = \mathcal{W}^2_{2}\left(\mu,\mathscr{P}^{cx}_{\preceq \nu}\right).
\]    
Hence,  $\mathcal{W}_{2}\left(\mu,\mathscr{P}^{cx}_{\preceq \nu}\right)=0$ if and only if $\mu \preceq \nu$. In this sense,  the \emph{backward} Wasserstein projection is equivalent to the weak optimal transport with barycentric cost.

\subsection{Dimensions of sets and measures}
In the following we briefly introduce upper and lower Wasserstein dimensions of measures. Roughly speaking, smaller Wasserstein dimensions of measure imply that it is concentrated on some low dimensional object. More details may be found \cite{JW_FB_sample_rates} and references therein.

\begin{definition}
The $\varepsilon$-covering number of a set $S$, denoted $\mathcal{N}%
_{\varepsilon}\left(  S\right)  $, is the minimum number of closed balls whose
union cover $S.$ The $\varepsilon$-dimension of $S$ is the quantity%
\[
d_{\varepsilon}\left(  S\right)  =\frac{\log\mathcal{N}_{\varepsilon}\left(
S\right)  }{-\log\mathcal{\varepsilon}}.
\]

Given a measure space $\left(  X,\mu\right)  ,$ the $\left(  \varepsilon
,\tau\right)  $-covering number of $X$ is defined as%
\[
\mathcal{N}_{\varepsilon}\left(  \mu,\tau\right)  =\inf\left\{  \mathcal{N}%
_{\varepsilon}\left(  S\right)  :\mu\left(  S\right)  \geq1-\tau\right\}  .
\]
The $\left(  \varepsilon,\tau\right)  $-dimension is the quantity%
\[
d_{\varepsilon}\left(  \mu,\tau\right)  =\frac{\log\mathcal{N}_{\varepsilon
}\left(  \mu,\tau\right)  }{-\log\mathcal{\varepsilon}}.
\]
\end{definition}

\begin{definition} The upper and lower Wasserstein dimensions are respectively,
\begin{align*}
    d^*_p(\mu) &:= \inf \left\{ s \in (2p, \infty) : \limsup_{\varepsilon \to 0} d_{\varepsilon}(\mu,\varepsilon^{\frac{sp}{s - 2p}}) \leq s \right\},\\
    d_*(\mu) &:= \lim_{\tau \to 0} \liminf_{\varepsilon \to 0} d_{\varepsilon}\left(  \mu,\tau\right).
\end{align*}
\end{definition}

\begin{remark}
In \cite{JW_FB_sample_rates}, the authors compare Wasserstein dimensions to the Minkowski and Hausdorff dimensions of measures, and get inequalities among them: see \cite[Proposition 2]{JW_FB_sample_rates}.
\end{remark}

\subsection{Consistent test}
We end this subsection with some definitions and terminologies about nonparametric hypothesis testing. Given a metric space $(\mathcal{F}, d)$, the model for the
distribution of samples consists of probability laws $\mathbb{P}_f$ indexed by $f \in \mathcal{F}$. A \emph{statistical hypothesis $\boldsymbol{H_0}$} is a subset of $\mathcal{F}$, and a \emph{statistical test} for $\boldsymbol{H_0}$ is a decision rule from samples to $\{0, 1\}$ where $0$ and $1$ denote \emph{accept} and \emph{reject} $\boldsymbol{H_0}$, respectively. In the following since we will have two sets of samples $\{X_1, \dots, X_n\}$ and $\{Y_1, \dots, Y_m\}$, double index $n,m$ is used to denote the dependency of samples.

We introduce a consistent test \cite[Definition 6.2.1]{gine2021mathematical}, which is the fundamental property required for hypothesis tests. In order to be consistent, the \emph{power} of a test increases to one as the number of samples grows to infinity.

\begin{definition}[Consistent test]
Consider the following hypothesis testing problem:
\[
    f \in \boldsymbol{H_0} \text{ vs } f \in \boldsymbol{H_A}.
\]
Let $\Psi_{n,m}$ be a function over the set of samples defined as
\[
    \Psi_{n,m} : \{X_1, \dots, X_n\} \times \{Y_1, \dots, Y_m\} \longrightarrow \{0, 1\},
\]
which is a statistical test for $\boldsymbol{H_0}$. $\Psi_{n,m}$ is said to be  \emph{consistent} if
\[
  \sup_{f \in \boldsymbol{H_0}} \mathbb{E}_{f} \Psi_{n,m} + \sup_{f \in 
 \boldsymbol{H_A}} \mathbb{E}_{f}\left[ 1 - \Psi_{n,m} \right] \leq \alpha_{n,m}
\]
for $\alpha_{n,m} \to 0$ as $n,m \to \infty$. 
\end{definition}

\section{Stability of the Wasserstein projection distances}\label{sec:stability}
In practical situations, we cannot access both of the backward projection $\mathcal{W}_{2}\left(
\mu,\mathscr{P}^{cx}_{\preceq \nu}\right)$ and the forward projection
$\mathcal{W}_{2}(  \mathscr{P}_{\mu \preceq}^{cx},\nu )$. Instead, we are given empirical distributions $\mu_{n}$, $\nu_m$ based on i.i.d. samples, and plug them in to
replace $\mu$, $\nu$. In other words, we can only access the approximation, or the estimator in statistics context, $\mathcal{W}_{2}( \mu_{n},\mathscr{P}^{cx}_{\preceq \nu_m})$ and $\mathcal{W}_{2}(  \mathscr{P}_{\mu_n \preceq}^{cx},\nu_m)$. One of the most fundamentally requisite properties for an estimator is \emph{consistency}. Namely, the hope is that if one accesses infinitely many samples, then the estimator will approach  the true quantity. This property is equivalent to \emph{stability} of functions in mathematics context. In our setting, if one can understand stability of the Wasserstein projections, then consistency of the estimator, and the control of the errors resulted from this approximation, could be quantitatively understood too. This section is devoted to answering this stability question, and proving the consistency result.

To understand the nature of the problem, we note that the stability of the Wasserstein projection involves more or less perturbing the convex order cones, $\mathscr{P}_{\mu \preceq}^{cx}$ and $\mathscr{P}^{cx}_{\preceq \nu}$. A quick thought seems to indicate e.g., that
$\mathcal{W}_{2}\left(  \mu,\mathscr{P}^{cx}_{\preceq \nu}\right)  $ should be stable to perturbation of $\mu$ because of the stability of the Wasserstein
distance, and that if $\mathscr{P}^{cx}_{\preceq \nu}$ is also stable with respect to
$\nu$ in a suitable sense, then the desired stability readily follows.
Although quite tempting, it is not straightforward to implement  this
intuition. It is mainly because that martingale property is fragile, as perturbing either one of the two marginals of a martingale pair can easily
breaks the martingale property. Also, consequently, the convex order cone is not stable as perturbing the vertex of the
convex order cone could produce larger-than-expected change of the cone. This
has already underlied the difficulty of showing the stability of martingale
optimal transport. In fact, the stability of martingale transport is expected
to hold only in one dimension \cite{backhoff2022stability}, and in higher
dimensions, with Euclidean transport cost, examples have been constructed
showing the instability of martingale optimal transport
\cite{bruckerhoff2022instability}. These facts indicate that the convex order cones $\mathscr{P}_{\mu \preceq}^{cx}$ and $\mathscr{P}^{cx}_{\preceq \nu}$ do
not possess the properties that readily lend themselves to the resolution of
the stability problem we are facing.

As a result, proving the stability of $\mathcal{W}_{2}(  \mu
,\mathscr{P}^{cx}_{\preceq \nu})  $ with respect to the vertex $\nu$ of the
convex order cone $\mathscr{P}^{cx}_{\preceq \nu}$ (or $\mathcal{W}_{2}(
\mathscr{P}_{\mu \preceq}^{cx},\nu)  $ with respect to the vertex $\mu$ of the
cone $\mathscr{P}_{\mu \preceq}^{cx}$) would not be trivial. Adding more challenge to the problem, a qualitative stability is not enough for our statistical analysis, which requires much finer quantitative estimates. On the other hand, the stability of interest is the weakest object: we want not the stability of optimal (martingale) coupling but the stability of the optimal value $\mathcal{W}_{2}(  \mu, \mathscr{P}^{cx}_{\preceq \nu})$. Despite of the apparent difficulty, the resolution is simple as we will see below, which is based on the understanding of the Wasserstein geometry of the convex order cone.

As we note in Section~\ref{sec: main results}, the stability of the Wasserstein projection measure in one dimension is studied in \cite{Jourdain_2023}, which is stronger than the stability of the Wasserstein projection distance. However, its proof relies on the quantile functions and the lattice structure of convex order, which hold only in one dimension. Due to these reasons, it seems non-trivial to extend their result to higher dimensions, for which we leave such stronger stability as a future work.

In the following, we will distinguish vertex and non-vertex marginals. In
backward Wasserstein projection $\mathcal{W}_{2}\left(  \mu, \mathscr{P}^{cx}_{\preceq \nu} \right)$, we refer to $\nu$ as vertex marginal and $\mu$
as non-vertex marginal. For forward Wasserstein projection\ $\mathcal{W}_{2} (  \mathscr{P}_{\mu \preceq}^{cx},\nu )  ,$\ $\mu$ is vertex marginal and $\nu$ is non-vertex marginal.

\subsection{Lipschtiz continuity with respect to non-vertex marginals}
To obtain the desired stability, we start with the following Lipschtiz
continuity of the Wasserstein projection with respect to non-vertex marginals.

The following lemma follows from the simple observation of metric spaces: if $x, y \in X$, a metric space, and $C \subset X$, then the metric distance should satisfy
$|d(x, C) - d(y, C)| \le d(x, y)$; to see this, notice that $d(x, C) \leq d(x, y) + d(y, C)$, and $d(y, C) \leq d(x, y) + d(x, C)$.

\begin{lemma}
\label{lem:cont_nvtex}
For any $\mu$, $\nu$, $\xi \in \mathcal{P}_2\left(\mathbb{R}^{d}\right)$, it holds for Wasserstein backward and forward projection that,%
\begin{align}\label{eq: cont_nvtex}
    &\left\vert \mathcal{W}_{2}\left(  \mu,\mathscr{P}^{cx}_{\preceq \nu}\right) -\mathcal{W}_{2}\left(  \xi,\mathscr{P}^{cx}_{\preceq \nu}\right)  \right\vert \leq \mathcal{W}_{2}\left(  \mu,\xi\right),\\ \nonumber
    & \left\vert \mathcal{W}_{2}\left(  \mathscr{P}_{\mu \preceq}^{cx},\nu\right)
-\mathcal{W}_{2}\left(  \mathscr{P}_{\mu \preceq}^{cx},\xi\right)  \right\vert
\leq \mathcal{W}_{2}\left(  \nu,\xi\right).
\end{align}
\end{lemma}

\subsection{Lipschtiz continuity with respect to vertex marginals}
The nontrivial stability with respect to the vertex-marginals, can in fact be proved by a simple method of using \eqref{eq: cont_nvtex} and  the geometric relation \eqref{eq : forward and backward are equal} between backward and forward Wasserstein projection.

\begin{lemma}\label{lem:cont_vtex}
For any $\mu$, $\nu$, $\xi\in \mathcal{P}_2\left(  \mathbb{R}^{d}\right)$, it holds for Wasserstein backward and forward projection that,
\begin{align}\label{eq: cont_vtex}
  &  \left\vert \mathcal{W}_{2}\left(  \mu,\mathscr{P}^{cx}_{\preceq \nu}\right)
-\mathcal{W}_{2}\left(  \mu,\mathscr{P}_{\preceq \xi}^{cx}\right)  \right\vert
\leq \mathcal{W}_{2}\left(  \nu,\xi\right),\\ \nonumber
& \left\vert \mathcal{W}_{2}\left(  \mathscr{P}_{\mu \preceq}^{cx},\nu\right) -\mathcal{W}_{2}\left(  \mathscr{P}_{\xi\leq }^{cx},\nu\right)  \right\vert \leq \mathcal{W}_{2}\left(  \mu,\xi\right) .
\end{align}
\end{lemma}

\begin{proof}
For the first line in \eqref{eq: cont_vtex}, that is, the Lipschitz continuity for backward projection,
observe that
\begin{align*}
    \mathcal{W}_{2}\left(  \mu,\mathscr{P}^{cx}_{\preceq \nu}\right)   &=\mathcal{W}_{2}\left(  \mathscr{P}_{\mu \preceq}^{cx},\nu\right)  \text{
by \eqref{eq : forward and backward are equal}}\\
    &  \leq \mathcal{W}_{2}\left(  \mathscr{P}_{\mu \preceq}^{cx},\xi\right) +\mathcal{W}_{2}\left(  \xi,\nu\right)  \text{ by \eqref{eq: cont_nvtex}}\\
    &  =\mathcal{W}_{2}\left(  \mu,\mathscr{P}_{\preceq \xi}^{cx}\right) +\mathcal{W}_{2}\left(  \xi,\nu\right)  \text{ by \eqref{eq : forward and backward are equal}.}
\end{align*}
This gives the inequality $\mathcal{W}_{2}\left(  \mu,\mathscr{P}^{cx}_{\preceq \nu}\right)
-\mathcal{W}_{2}\left(  \mu,\mathscr{P}_{\preceq \xi}^{cx}\right)
\leq \mathcal{W}_{2}\left(  \nu,\xi\right)$. Since  the same inequality holds with exchanging $\nu$ and $\xi$, this completes the proof of the Lipschitz continuity for backward projection.

The Lipschitz continuity of the forward projection can be proved with the same method, and we skip it.
\end{proof}

\subsection{Quantitative stability}
Combining what we have so far, we obtain the stability of the Wasserstein projections with respect to both marginals, one of the main results in our paper.

\begin{theorem}[Quantitative stability]\label{thm: consistency of projected Wasserstein-1}
For any probability measures $\mu, \mu'$ and $\nu, \nu'$, 
\begin{align}\label{eq : upper bound of W_2-1}
    \left|\mathcal{W}_{2}\left(  \mu',\mathscr{P}_{\preceq \nu'}^{cx}\right) - \mathcal{W}_{2}\left(
    \mu,\mathscr{P}^{cx}_{\preceq \nu}\right)\right| &\leq    \mathcal{W}_{2}\left( \mu, \mu' \right)  + \mathcal{W}_{2}\left(  \nu,\nu'\right).
\end{align}
\end{theorem}

\begin{proof}
Observe that
    \begin{align*}
   & \left|\mathcal{W}_{2}\left(  \mu',\mathscr{P}_{\preceq \nu'}^{cx}\right) - \mathcal{W}_{2}\left(
    \mu,\mathscr{P}^{cx}_{\preceq \nu}\right)\right| \\
   &\leq  \left|\mathcal{W}_{2}\left(  \mu',\mathscr{P}_{\preceq \nu'}^{cx}\right) - \mathcal{W}_{2}\left(  \mu,\mathscr{P}_{\preceq \nu'}^{cx}\right) + \mathcal{W}_{2}\left(  \mu,\mathscr{P}_{\preceq \nu'}^{cx}\right)  - \mathcal{W}_{2}\left(
    \mu,\mathscr{P}^{cx}_{\preceq \nu}\right)\right| \\
    & \leq  \left|\mathcal{W}_{2}\left(  \mu',\mathscr{P}_{\preceq \nu'}^{cx}\right) - \mathcal{W}_{2}\left(  \mu,\mathscr{P}_{\preceq \nu'}^{cx}\right)\right| + \left| \mathcal{W}_{2}\left(  \mu,\mathscr{P}_{\preceq \nu'}^{cx}\right)  - \mathcal{W}_{2}\left(
    \mu,\mathscr{P}^{cx}_{\preceq \nu}\right)\right|.
\end{align*}
Applying \eqref{eq: cont_nvtex} and \eqref{eq: cont_vtex} for the first term and the second term in the above last line, respectively, the conclusion follows.
\end{proof}

\begin{remark}[$\mathcal{W}_p$ case]\label{rem:Stable-Wp}
As mentioned in Remark~\ref{rem:Wp-swap},
the backward-forward swapping 
property \eqref{eq : forward and backward are equal} of Theorem~\ref{thm:swap_prop} holds for $\mathcal{W}_p$ with appropriate modification. Then, Lemma~\ref{lem:cont_vtex} and Theorem~\ref{thm: consistency of projected Wasserstein-1} can be extended to $\mathcal{W}_p$ for $1\leq p <\infty$ with exactly the same proof. In what follows, we will simply focus on the $\mathcal{W}_2$ case since it is sufficient for our statistical application. 
\end{remark}

\section{Convergence rate of $\mathcal{W}_{2}\left(  \mu_n,\mathscr{P}_{\preceq \nu_m}^{cx}\right)$}\label{sec : Convergence rate}
In this section, we tackle the question of whether the empirical projection distance converges to the true projection distance, and if it does at what rate is the convergence? The latter question is related to the rate of convergence of the empirical measure in the classical Wasserstein distance.

Understanding the rate of convergence of empirical measures in Wasserstein distance has been an active research field due to the applicability of Wasserstein geometry in various scientific areas. In his seminal paper \cite{dudley69} Dudley paved the way of understanding this problem, but the proof works for only $\mathcal{W}_1$. Since then, there have been numerous papers to obtain the rate of convergence in Wasserstein distance beyond $\mathcal{W}_1$: see \cite{MR2280433, MR2861675, Dereich_Scheutzow_Schottstedt2013, MR3189084, NF_AG_rate_Wasserstein, MR4359822, merigot2021non, MR4255123, MR4441130}. Lower bounds established in \cite{JW_FB_sample_rates} indicate that in general, the rate of convergence of the usual empirical distribution suffers the curse of dimensionality, which says that the rate of convergence is $O\left( n^{-1/d} \right)$. This rate is inevitable unless a lower dimensional structure is assumed.

In this paper we particularly focus on two families of distributions, one satisfying \emph{log-Sobolev inequality} and the other having \emph{bounded support}, and utilize the results from \cite{NF_AG_rate_Wasserstein, JW_FB_sample_rates}. In \cite{NF_AG_rate_Wasserstein}, the authors prove the upper bound of the convergence rate of Wasserstein distance of empirical measures with a higher finite moment assumption, and concentration inequalities. The authors of \cite{JW_FB_sample_rates} prove the rate of convergence in $p$-Wasserstein distance of empirical measures and a sub-gaussian type concentration inequality with the assumption of bounded support. In particular, if the support of $\mu$ has an intrinsic lower dimension, \cite[Theorem 1]{JW_FB_sample_rates} provides the faster convergence rate of Wasserstein distance of empirical measures, relying on the \emph{Wasserstein dimension}. 

Recall that a probability measure $\mu$ on $\mathbb{R}^d$ is said to satisfy the log-Sobolev inequality with constant $\kappa>0$ if for any smooth function $f$,
\begin{equation}\label{eq: log-sobolev definition}
    \operatorname{Ent}_\mu(f^2) \leq \kappa \int_{\mathbb{R}^d} \big|\nabla f(x)\big|^2\,d\mu(x),
\end{equation}
where $\operatorname{Ent}_\mu(f^2) = \int_{\mathbb{R}^d} f^2(x)\log \left( \frac{f^2(x)}{\int_{\mathbb{R}^d}f^2(x)\,d\mu(x)} \right) d\mu(x)$ is the entropy functional.

In the following, we will derive desirable concentrations for Wasserstein projection. In the case of  Wasserstein distance, sharp concentrations for log-Sobolev and bounded-support measures are respectively developed in \cite{lei2020convergence} and \cite{JW_FB_sample_rates}; see also references therein. For Wasserstein projections, the following concentration inequalities are new and they are expected to play equally important roles when statistical inference is involved.

\begin{lemma}[Concentration inequality]\label{lem: concentration}
Let $\mu, \nu \in \mathcal{P}_2 \left( \mathbb{R}^d \right)$.
\begin{enumerate}
    \item[(i)] If $\mu$ and $\nu$ satisfy log-Sobolev inequality \eqref{eq: log-sobolev definition} with some constant $\kappa>0$, then for $t\geq 0$,
\begin{equation}\label{eq: concentrate_log_Sob}
    \mathbb{P}\left(  \left\vert \mathcal{W}_{2}\left(  \mu_n,\mathscr{P}_{\preceq \nu_m}^{\text{cx}}\right)  -\mathbb{E} \left[ \mathcal{W}_{2}\left(  \mu_n,\mathscr{P}_{\preceq \nu_m}^{\text{cx}}\right)\right] \right\vert \geq t\right)  \leq  2\exp\left(  -\frac{\left(  m\wedge n\right)  t^{2}}{2\kappa}\right).
\end{equation}
    \item[(ii)] If $\mu$ and $\nu$ are supported with bounded supports of diameter $D > 0$, then for $t \geq 0$,
\begin{equation}\label{eq: concentrate_bd_supp}
    \mathbb{P}\left(  \left \vert \mathcal{W}_{2}\left(  \mu_{n},\mathscr{P}_{\preceq\nu_{m}}^{\text{cx}}\right)  -\mathbb{E}\left[  \mathcal{W}_{2}\left(\mu_{n},\mathscr{P}_{\preceq\nu_{m}}^{\text{cx}}\right)  \right] \right\vert \geq t\right)  \leq 2\exp\left(  -\frac{\left(  m\wedge n\right)  ^{2}t^{4}}{2\left(  m+n\right)  D^{4}}\right)  .
\end{equation}
\end{enumerate}
\end{lemma}
The proof of this lemma is involved and is given in the \Cref{app: sec : Convergence rate}.

\begin{remark}
By the backward-forward swap property \eqref{eq : forward and backward are equal}, the forward Wasserstein projection
$\mathcal{W}_{2}\left( \mathscr{P}_{\mu_{n}\preceq}^{\text{cx}},\nu
_{m}\right)$ enjoys the same concentration inequalities.
\end{remark}

The next is an upper bound of the expectation of the Wasserstein distance of the empirical measures with finite higher moments and bounded support, respectively, developed in \cite[Theorems 1]{NF_AG_rate_Wasserstein} and \cite[Proposition 5]{JW_FB_sample_rates}. Although \cite[Theorems 1]{NF_AG_rate_Wasserstein} does not describe the exact constant, it is not too hard to obtain it following their proof step by step. Here we state their result for $p=2$ only.

\begin{lemma}[Upper bound of the expected empirical Wasserstein distance]\cite[Theorems 1]{NF_AG_rate_Wasserstein}\cite[Proposition 5]{JW_FB_sample_rates} \label{thm : upper bound of expected distance}
\begin{enumerate}
    \item[(i)] Let $\mu$ be a probability measure on $\mathbb{R}^{d}$ with a finite $5$-th moment, $M_5(\mu):= \int_{\mathbb{R}^{d}}\Vert x\Vert^{5}_5 d\mu(x)$, and let $\mu_{n}$ be an empirical measure for $\mu$. Then, for all $n \geq 1$,
\begin{equation}\label{eq: convergence rate of FG15}
    \begin{aligned}
    \mathbb{E}[\mathcal{W}_2(\mu,\mu_n)] \leq F(\mu, n ,d) := 20 d M^{\frac{1}{5}}_5(\mu) \times
    \left\{ 
    \begin{array}{ll}
    n^{-\frac{1}{4}} & \textrm{if $d=1, 2, 3$,}\\
    n^{-\frac{1}{4}}\sqrt{\log (1 + n)} & \textrm{if $d=4$,}\\
    n^{-\frac{1}{d}} & \textrm{if $d \geq 5$.}
    \end{array} 
    \right.
    \end{aligned}
\end{equation}
    \item[(ii)] Let $\mu$ be a probability measure on $\mathbb{R}^{d}$ with a bounded support with diameter $D$. If $k > d^*_2(\mu) \vee 4$, then there exists constant $C=C(k) > 0$ such that
\begin{equation}\label{eq: convergence rate of WB19}
    \mathbb{E}[ \mathcal{W}_2(\mu, \mu_n)] \leq G(\mu, n, k):= \left( D^2  3^{\frac{12k}{k - 4} + 1}\left( \frac{1}{3^{\frac{k}{2} - 2} - 1} + 3 \right) n^{- \frac{2}{k}} + D^2 C^{\frac{k}{2}} n^{- \frac{1}{2}} \right)^{\frac{1}{2}}
\end{equation}
where $d^*_p(\mu)$ is the upper p-Wasserstein dimension.
\end{enumerate}
\end{lemma}

\begin{remark}
In general, if a probability measure satisfies the log-Sobolev inequality, then it has finite moments: indeed, it is known \cite[Introduction]{ledoux1999concentration} that if a probability measure $\mu$ satisfies the log-Sobolev inequality, then $\mu$ necessarily satisfies the stronger property of exponential integrability: $\int e^{\alpha \Vert x \Vert^2} d\mu(x)<\infty$ for some $\alpha>0$.  Hence, the precondition of Lemma \ref{thm : upper bound of expected distance} (i) is automatically satisfied if \eqref{eq: log-sobolev definition} holds for some $\kappa > 0$. On the other hand, it is possible to give concrete sufficient criteria for a measure to satisfy the log-Sobolev inequality, for example by way of the Holley-Stroock theorem \cite[p. 200]{ledoux1999concentration}.
\end{remark}

\begin{remark}
Note that in \eqref{eq: convergence rate of WB19} the constant of $n^{-\frac{2}{k}}$ decreases as $k$ increases as long as $k > 4$ though the constant $C(k)$ depends on $k$ exponentially. If $d^*_2(\mu) \leq 4$, then one can choose $k > 4$ arbitrarily to minimize the RHS of \eqref{eq: convergence rate of WB19}.

If the support of $\mu$ is a regular set of dimension $d$ and $\mu \ll \mathcal{H}^d$, the d-dimensional Hausdorff measure, then $d^*_p(\mu) = d$ for any $p \in [1, \infty)$. For example, if $\mu$ is absolutely continuous with respect to Lebesgue measure, then $d^*_p(\mu) = d$. See \cite[Propositions 8 and 9]{JW_FB_sample_rates} for more details.    
\end{remark}

The following proposition describes the tail probability of $\mathcal{W}_{2}\left(  \mu_n,\mathscr{P}_{\preceq \nu_m}^{cx}\right)$. It is valid when $t$ vanishes not too fast in terms of $n$ and $m$. Let $M > 0$ be fixed such that $M_5(\mu), M_5(\nu) \leq M$. Recall $F(\mu, n, d)$ and $G(\mu, n ,k)$ from \eqref{eq: convergence rate of FG15} and \eqref{eq: convergence rate of WB19}, respectively.

\begin{proposition}[Tail probability]\label{prop: tail_probability}
Let $\mu, \nu \in \mathcal{P}_2 \left( \mathbb{R}^d \right)$.
\begin{enumerate}
    \item[(i)] Assume that $\mu$ and $\nu$ satisfy the log-Sobolev inequality \eqref{eq: log-sobolev definition} with a constant $\kappa > 0$ and $M_5(\mu), M_5(\nu) \leq M$. Suppose $t$ satisfies 
\begin{equation}\label{eq : sample bound with respect to t}
    t > 4 \left( F(\mu, n, d) \vee F(\nu, m, d) \right).
\end{equation}
Then, 
\begin{equation*}\label{eq: p-value log sobolev}
    \mathbb{P}\left( \vert \mathcal{W}_{2}\left(   \mu_n,\mathscr{P}_{\preceq \nu_m}^{cx}\right) - \mathcal{W}_{2}\left(  \mu,\mathscr{P}_{\preceq \nu}^{cx}\right) \vert \geq t  \right)
    \leq  2\exp\left(  -\frac{\left(  m\wedge n\right)  t^{2}}{8\kappa}\right).
\end{equation*}

    \item[(ii)] Assume that $\mu$ and $\nu$ have bounded supports with diameter at most $D$ and $k > d^*_2(\mu) \vee d^*_2(\nu) \vee 4$. Suppose $t$ satisfies 
\begin{equation}\label{eq : sample bound with respect to t bounded support}
    t > 4 \left( G(\mu, n ,k) \vee G(\nu, m ,k) \right).
\end{equation}
Then,  
\begin{equation*}\label{eq: p-value bounded support}
    \mathbb{P}\left(  \vert \mathcal{W}_{2}\left(  \mu_n,\mathscr{P}_{\preceq \nu_m}^{cx}\right) - \mathcal{W}_{2}\left(  \mu,\mathscr{P}_{\preceq \nu}^{cx}\right) \vert \geq t   \right) \leq  2\exp\left(  -\frac{ \left(  m\wedge n\right)^{2}t^{4}}{32 \left(  m+n\right)  D^{4}}\right) 
\end{equation*}
\end{enumerate}

\end{proposition}
The proof of this proposition is involved and is given in the \Cref{app: sec : Convergence rate}.

Now, we are ready to describe 
the convergence rate of $\mathcal{W}_{2}\left(  \mu_n,\mathscr{P}_{\preceq \nu_m}^{cx}\right)$, which is straightforward by \Cref{prop: tail_probability}. We omit the proof.

\begin{theorem}[Convergence rate]\label{thm : Convergence rate}
Let $\mu, \nu \in \mathcal{P}_2 \left( \mathbb{R}^d \right)$.
\begin{enumerate}
    \item[(i)] Assume that $\mu$ and $\nu$ satisfy the log-Sobolev inequality \eqref{eq: log-sobolev definition} with some constant $\kappa > 0$ and $M_5(\mu), M_5(\nu) \leq M$. Then, 
\begin{equation}\label{eq: convergence rate under moment log-Sobolev}
\begin{aligned}
    \left| \mathcal{W}_{2}\left(  \mu_n,\mathscr{P}_{\preceq \nu_m}^{cx}\right) -  \mathcal{W}_2(\mu, \mathscr{P}^{cx}_{\preceq \nu}) \right| < 4 \left( F(\mu, n, d) \vee F(\nu, m, d) \right)
\end{aligned}   
\end{equation}
with probability $1 - 2\exp\left(  -\frac{\left(  m\wedge n\right)  t^{2}}{8\kappa}\right)$.
    \item[(ii)] Assume that $\mu$ and $\nu$ have bounded supports with diameter at most $D$. If $k > d^*_2(\mu) \vee d^*_2(\nu) \vee 4$, then 
\begin{equation}\label{eq: convergence rate under bounded support}
\begin{aligned}
    \left| \mathcal{W}_{2}\left(  \mu_n,\mathscr{P}_{\preceq \nu_m}^{cx}\right) - \mathcal{W}_{2}\left(  \mu,\mathscr{P}_{\preceq \nu}^{cx}\right) \right| < 4 \left( G(\mu, n ,k) \vee G(\nu, m ,k) \right)
\end{aligned}   
\end{equation}
with probability $1 - 2\exp\left(  -\frac{ \left(  m\wedge n\right)^{2}t^{4}}{32 \left(  m+n\right)  D^{4}}\right)$.
\end{enumerate}
\end{theorem}



\begin{remark}
\eqref{eq: convergence rate under moment log-Sobolev} says that the rate of convergence for log-Sobolev case is roughly 
\[
    \left| \mathcal{W}_{2}\left(  \mu_n,\mathscr{P}_{\preceq \nu_m}^{cx}\right) -  \mathcal{W}_2(\mu, \mathscr{P}^{cx}_{\preceq \nu}) \right| \leq O((n \wedge m))^{-\frac{1}{d}}.
\]
Similarly, \eqref{eq: convergence rate under bounded support} says that the rate of convergence for the bounded case is roughly 
\[
    \left| \mathcal{W}_{2}\left(  \mu_n,\mathscr{P}_{\preceq \nu_m}^{cx}\right) -  \mathcal{W}_2(\mu, \mathscr{P}^{cx}_{\preceq \nu}) \right| \leq O((n \wedge m))^{-\frac{1}{k}}.
\]

There are pros and cons of \eqref{eq: convergence rate under moment log-Sobolev} and \eqref{eq: convergence rate under bounded support}. Although \eqref{eq: convergence rate under bounded support} has a better(smaller) dependency of the dimension than \eqref{eq: convergence rate under moment log-Sobolev} does, generally it is a hard task to estimate the upper Wassertein dimension of a measure, $d^*_p(\mu)$, in reality due to its fractal spirit. Notice that if $\mu$ has a bounded support, then it automatically satisfies log-Sobolev inequality with a constant proportional to the diameter. Hence, if $k$ is unknown, \eqref{eq: convergence rate under moment log-Sobolev} works for the bounded support case.

About the estimation of the intrinsic dimension, there is a recent paper \cite{JMLR:v23:21-1483} in which the authors suggest a new method to estimate it based on \cite{JW_FB_sample_rates}. We also refer the readers to \cite{camastra2016intrinsic} for the history about the intrinsic dimension estimation problem.
\end{remark}



\section{\emph{p}-value, type I/II errors and consistency}\label{sec: p-value, errors}

In this section, we consider upper bounds of the $p$-value and type I/II errors of $\mathcal{W}_{2}\left(  \mu_n,\mathscr{P}_{\preceq \nu_m}^{cx}\right)$. Given an estimator $T$, the $p$-value for (a one-sided right-tail) test-statistic distribution is $\mathbb{P}\left( T \geq t \,|\, \boldsymbol{H_0} \right)$. 
The $p$-value can be computed 
given the exact form of the distribution of $T$. However, in our case we only know the tail probability of $\mathcal{W}_{2}\left(  \mu_n,\mathscr{P}_{\preceq \nu_m}^{cx}\right)$.


Thanks to the fact that \Cref{prop: tail_probability} only depends on either log-Sobolev constant and $5$-th moment or the upper Wasserstein dimension and the diameter of the support, we can control $p$-value and type I/II errors uniformly over  two proper families of probability measures, which are big enough to cover most practical situations. A lower bound of the critical value $t(\alpha)$ for \eqref{eqn:test-new-t} can be derived too in the same way. 

Before we start, we define
\[
       \mathcal{P}_{\preceq}:=\left\{ (\rho_1, \rho_2) \in \mathcal{P}_2(\mathbb{R}^d)\times \mathcal{P}_2(\mathbb{R}^d) : \mathcal{W}_{2}\left(  \rho_1,\mathscr{P}_{\preceq \rho_2}^{cx}\right) = 0 \right\},
\]
the set of $(\rho_1, \rho_2)$ in convex order. In a similar fashion, for $r > 0$ define
\[
    \mathcal{P}^c_{\preceq}(r):=\left\{ (\rho_1, \rho_2)  \in \mathcal{P}_2(\mathbb{R}^d) \times \mathcal{P}_2(\mathbb{R}^d) : \mathcal{W}_{2}\left(  \rho_1,\mathscr{P}_{\preceq \rho_2}^{cx}\right)  >  r \right\},
\]
the set of $(\rho_1, \rho_2)$ not in convex order, with a strict separation $r>0$.



In order to show that the proposed test \eqref{eqn:test-new-t} is consistent, we consider
two families which rise from the log-Sobolev inequality case and the bounded support case: Define 
\begin{align*}
    \mathcal{P}_{\kappa, M} := \{
     (\rho_1, \rho_2)  \in \mathcal{P}_2(\mathbb{R}^d) \times \mathcal{P}_2(\mathbb{R}^d)
      : \ \ &\hbox{$M_5(\rho_1) \vee M_5(\rho_2) \leq M$,}\\
    &\hbox{and $\rho_1$ and $\rho_2$ satisfy \eqref{eq: log-sobolev definition} with $\kappa$}\}, 
\end{align*}
for fixed $\kappa, M>0$, and
\begin{align*}
   \mathcal{Q}_{k, D} := \{ 
    (\rho_1, \rho_2)  \in \mathcal{P}_2(\mathbb{R}^d) \times \mathcal{P}_2(\mathbb{R}^d)
      : \ \ 
   &\text{diam}(\spt(\rho_1)) \vee  \text{diam} (\spt(\rho_2)) \leq D\\
   &\text{ and } k > d^*_2(\rho_1) \vee d^*_2(\rho_2) \vee 4\}
\end{align*}
for fixed $k, D > 0$. In the following, $\mathcal{P}_{\preceq}$ will be restricted to either $\mathcal{P}_{\kappa, M}$ or $\mathcal{Q}_{k, D}$ to control type I and II errors uniformly. For notational simplicity, let us introduce
\begin{equation*}\label{eq: restriction 1}
    \mathcal{P}_{\preceq, \kappa, M}:= \mathcal{P}_{\kappa, M} \cap \mathcal{P}_{\preceq}, \quad  \mathcal{Q}_{\preceq, k, D}:= \mathcal{Q}_{k, D} \cap \mathcal{P}_{\preceq}
\end{equation*}
and similarly
\begin{equation*}\label{eq: restriction 2}
    \mathcal{P}^c_{\preceq, \kappa, M}(r):= \mathcal{P}_{\kappa, M} \cap \mathcal{P}^c_{\preceq}(r), \quad \mathcal{Q}^c_{\preceq, k, D}(r):= \mathcal{Q}_{k, D} \cap \mathcal{P}^c_{\preceq}(r).
\end{equation*}
Now the following proposition is a direct consequence of \Cref{prop: tail_probability}, and we omit its proof.

\begin{proposition}[$p$-value, Type I error]\label{prop: p-value, Type I error}
Let $\mu, \nu \in \mathcal{P}_2 \left( \mathbb{R}^d \right)$.
\begin{enumerate}
    \item[(i)] Assume that $\mu$ and $\nu$ satisfy the log-Sobolev inequality \eqref{eq: log-sobolev definition} with a constant $\kappa > 0$ and $M_5(\mu), M_5(\nu) \leq M$. Let $\boldsymbol{H_0} : \mathcal{P}_{\preceq, \kappa, M}$. If $t$ satisfies \eqref{eq : sample bound with respect to t}, then
\[
     \sup_{\mathcal{P}_{\preceq, \kappa, M}} \mathbb{P}\left( \mathcal{W}_{2}\left(  \mu_n,\mathscr{P}_{\preceq \nu_m}^{cx}\right) \geq t  \right) \leq 2\exp\left(  -\frac{\left(  m\wedge n\right)  t^{2}}{8\kappa}\right).
\]
Define
\begin{equation}\label{eq: critical t log-sobolev}
    t_1(\alpha) :=  \left\{ 4 \left( F(\mu, n, d) \vee F(\nu, m, d) \right) \right\} \vee \left( \frac{8 \kappa \log \frac{2}{\alpha}}{n \wedge m} \right)^{\frac{1}{2}}.
\end{equation}
If $t \geq t_1(\alpha)$, then
\[
     \sup_{\mathcal{P}_{\preceq, \kappa, M}} \mathbb{P}\left( \mathcal{W}_{2}\left(  \mu_n,\mathscr{P}_{\preceq \nu_m}^{cx}\right) \geq t  \right) \leq \alpha.
\]

    \item[(ii)] Assume that $\mu$ and $\nu$ have bounded supports with diameter at most $D$ and $d^*_2(\mu) \vee d^*_2(\nu) \vee 4 < k$. Let $\boldsymbol{H_0} : \mathcal{Q}_{\preceq, k, D}$. If $t$ satisfies \eqref{eq : sample bound with respect to t bounded support}, then
\[
     \sup_{\mathcal{Q}_{\preceq, k, D}}\mathbb{P}\left( \mathcal{W}_{2}\left(  \mu_n,\mathscr{P}_{\preceq \nu_m}^{cx}\right) \geq t \right) \leq  2\exp\left(  -\frac{ \left(  m\wedge n\right)^{2}t^{4}}{32 \left(  m+n\right)  D^{4}}\right) .
\] 
Define
\begin{equation}\label{eq: critical t bounded support}
    t_2(\alpha) :=  \left\{ 4 \left( G(\mu, n ,k) \vee G(\nu, m ,k) \right) \right\} \vee \left( \frac{32 (n+m) D^4 \log \frac{2}{\alpha}}{n \wedge m} \right)^{\frac{1}{4}}
\end{equation}
If $t \geq t_2(\alpha)$, then
\[
     \sup_{\mathcal{Q}_{\preceq, k, D}}\mathbb{P}\left( \mathcal{W}_{2}\left(  \mu_n,\mathscr{P}_{\preceq \nu_m}^{cx}\right) \geq t  \right) \leq \alpha.
\]
\end{enumerate}
\end{proposition}

Based on the derivation of a lower bound of the critical value $t(\alpha)$, one can also compute type II error, which is the probability of failing to reject $\boldsymbol{H_0}$ under $\boldsymbol{H_A}$. In our setting, the type II error is given by
\[
    \sup_{\mathcal{P}^c_{\preceq}(r)}\mathbb{P}\left( \mathcal{W}_{2}\left(  \mu_n,\mathscr{P}_{\preceq \nu_m}^{cx}\right) < t  \right).
\]
Before estimating this, it is worthwhile to justify the necessity of the consideration of $\mathcal{P}^c_{\preceq}(r)$ with a positive $r$. Assume that $\boldsymbol{H_A} : \mathcal{P}^c_{\preceq}(0)$, i.e., no strict separation between $\boldsymbol{H_0}$ and $\boldsymbol{H_A}$. Observe that
\begin{align*}
    &\sup_{\mathcal{P}^c_{\preceq}(0)}\mathbb{P}\left( \mathcal{W}_{2}\left(  \mu_n,\mathscr{P}_{\preceq \nu_m}^{cx}\right) < t  \right)\\
    &= \sup_{\mathcal{P}^c_{\preceq}(0)}\mathbb{P}\left(  \mathcal{W}_{2}\left(  \mu_n,\mathscr{P}_{\preceq \nu_m}^{cx}\right) -\mathcal{W}_{2}\left(
    \mu,\mathscr{P}^{cx}_{\preceq \nu}\right)  < t - \mathcal{W}_{2}\left(
    \mu,\mathscr{P}^{cx}_{\preceq \nu}\right) \right).
\end{align*}
The problem is that, while we know that the value of 
$\mathcal{W}_{2}\left(  \mu_n,\mathscr{P}_{\preceq \nu_m}^{cx}\right)$ is close to that of $\mathcal{W}_{2}\left(\mu,\mathscr{P}^{cx}_{\preceq \nu}\right)$, we do not know the latter 
(we merely know that it is strictly positive under the alternative though). Thus, we cannot uniformly bound the Type II error over $\mathcal{P}^c_{\preceq}(0)$.

To overcome this issue, we need to impose some separation (with respect to the projection distance) between the null and the alternative hypotheses. Accordingly, let us consider the \emph{strict alternative} 
\[
    \boldsymbol{H_A}(\alpha): \mathcal{P}^c_{\preceq}(2t(\alpha))
\]
where $t(\alpha)$ is either $t_1(\alpha)$ or $t_2(\alpha)$ given in Proposition~\ref{prop: p-value, Type I error}. It establishes that there is, in some sense, $t(\alpha)$-separation between $\mathcal{P}_{\preceq}$ and $\mathcal{P}^c_{\preceq}(2t(\alpha))$. With the separation, type I/II errors can be controlled simultaneously.

Continuing the previous calculation, it follows that for any $t \leq t(\alpha)$
\begin{align*}
    &\sup_{\mathcal{P}^c_{\preceq}(2t(\alpha))}\mathbb{P}\left(\mathcal{W}_{2}\left(\mu_{n},\mathscr{P}_{\preceq\nu_{m}}^{cx}\right) < t \right)\\
    &\leq \sup_{\mathcal{P}^c_{\preceq}(2t(\alpha))}\mathbb{P}\left(\mathcal{W}_{2}\left(\mu_{n},\mathscr{P}_{\preceq\nu_{m}}^{cx}\right) < t(\alpha)  \right)\\
    &\leq \sup_{\mathcal{P}^c_{\preceq}(2t(\alpha))}\mathbb{P}\left( \mathcal{W}_{2}\left(\mu_{n},\mathscr{P}_{\preceq\nu_{m}}^{cx}\right) - \mathcal{W}_{2}\left(
    \mu,\mathscr{P}^{cx}_{\preceq \nu}\right) < t(\alpha) - \mathcal{W}_{2}\left(
    \mu,\mathscr{P}^{cx}_{\preceq \nu}\right)  \right)\\
    &= \sup_{\mathcal{P}^c_{\preceq}(2t(\alpha))}\mathbb{P}\left( \vert \mathcal{W}_{2}\left(\mu_{n},\mathscr{P}_{\preceq\nu_{m}}^{cx}\right) - \mathcal{W}_{2}\left(
    \mu,\mathscr{P}^{cx}_{\preceq \nu}\right) \vert  > \vert t(\alpha) - \mathcal{W}_{2}\left(
    \mu,\mathscr{P}^{cx}_{\preceq \nu}\right)  \vert \right)\\
    &\leq  \sup_{\mathcal{P}^c_{\preceq}(2t(\alpha))}\mathbb{P}\left( \vert \mathcal{W}_{2}\left(\mu_{n},\mathscr{P}_{\preceq\nu_{m}}^{cx}\right) - \mathcal{W}_{2}\left(
    \mu,\mathscr{P}^{cx}_{\preceq \nu}\right) \vert  >  t(\alpha) \right)
\end{align*}
where the last inequality follows from the definition of $\mathcal{P}^c_{\preceq}(2t(\alpha))$. Now applying \Cref{prop: tail_probability}, a routine calculation gives us an upper bound of the type II error for $\mathcal{W}_{2}\left(\mu_{n},\mathscr{P}_{\preceq\nu_{m}}^{cx}\right)$: 

\begin{proposition}[Type II error]\label{prop: type II error}
Let $\mu, \nu \in \mathcal{P}_2 \left( \mathbb{R}^d \right)$.
\begin{enumerate}
    \item[(i)] Assume that $\mu$ and $\nu$ satisfy the log-Sobolev inequality \eqref{eq: log-sobolev definition} with a constant $\kappa > 0$ and $M_5(\mu), M_5(\nu) \leq M$ and $t_1(\alpha)$ is given as \eqref{eq: critical t log-sobolev}. Let $\boldsymbol{H_A}(\alpha) :  \mathcal{P}^c_{\preceq, \kappa, M}(2t_1(\alpha))$. If $t \leq t_1(\alpha)$, then
\[
    \sup_{\mathcal{P}^c_{\preceq, \kappa, M}(2t_1(\alpha))}\mathbb{P}\left(\mathcal{W}_{2}\left(\mu_{n},\mathscr{P}_{\preceq\nu_{m}}^{cx}\right)< t \right) \leq \alpha.
\]

    \item[(ii)] Assume that $\mu$ and $\nu$ have bounded supports with diameter at most $D$ and $d^*_2(\mu) \vee d^*_2(\nu) \vee 4 < k$ and $t_2(\alpha)$ is given as \eqref{eq: critical t bounded support}. Let $\boldsymbol{H_A}(\alpha) : \mathcal{Q}^c_{\preceq, k, D}(2t_2(\alpha))$. If $t \leq t_2(\alpha)$, then
\[
    \sup_{\mathcal{Q}^c_{\preceq, k, D}(2t_2(\alpha))}\mathbb{P}\left(\mathcal{W}_{2}\left(\mu_{n},\mathscr{P}_{\preceq\nu_{m}}^{cx}\right)< t \right) \leq \alpha.
\]
\end{enumerate}     
\end{proposition}


Since $t(\alpha)$ is proportional to $\log \frac{1}{\alpha}$,  a routine calculation combining Propositions~\ref{prop: p-value, Type I error} and \ref{prop: type II error} will give  the consistency of \eqref{eqn:test-new-t}, as long as $\alpha= o\left(\exp(-(n \wedge m)) \right)$:

\begin{theorem}[Consistency]\label{thm: consistency of the test}
Let $\alpha_{n,m}:= \exp(-(n\wedge m)^s )$ for some $0 < s < 1$.
\begin{enumerate}
    \item[(i)] Assume that $\mu$ and $\nu$ satisfy the log-Sobolev inequality \eqref{eq: log-sobolev definition} with a constant $\kappa > 0$ and $M_5(\mu), M_5(\nu) \leq M$ and $t_1(\alpha)$ is given as \eqref{eq: critical t log-sobolev}.  
Consider
\[
    \boldsymbol{H_0} : \mathcal{P}_{\preceq, \kappa, M} \text{ vs } \boldsymbol{H_A}(\alpha_{n,m}) : \mathcal{P}^c_{\preceq, \kappa, M}(2t_1(\alpha_{n,m}/2)).
\]
Then, 
\begin{align*}  
 \alpha_{n,m}
    \geq & \  \sup_{\mathcal{P}_{\preceq, \kappa, M}} \mathbb{P}\left( \mathcal{W}_{2}\left(  \mu_n,\mathscr{P}_{\preceq \nu_m}^{cx}\right) \geq t_1(\alpha_{n,m}/2)  \right)\\
   &  \quad  + \sup_{\mathcal{P}^c_{\preceq, \kappa, M}(2t_1(\alpha_{n,m}/2))}\mathbb{P}\left(\mathcal{W}_{2}\left(\mu_{n},\mathscr{P}_{\preceq\nu_{m}}^{cx}\right)< t_1(\alpha_{n,m}/2)   \right). 
\end{align*}
    \item[(ii)] Assume that $\mu$ and $\nu$ have bounded supports with diameter at most $D$, $d^*_2(\mu) \vee d^*_2(\nu) \vee 4 < k$ and $t_2(\alpha)$ is given as \eqref{eq: critical t bounded support}. 
Consider
\[
    \boldsymbol{H_0} : \mathcal{Q}_{\preceq, k, D} \text{ vs } \boldsymbol{H_A}(\alpha_{n,m}) : \mathcal{Q}^c_{\preceq, k, D}(2t_2(\alpha_{n,m}/2)).
\]
Then,
\begin{align*}
 \alpha_{n,m} 
   \geq & \  \sup_{\mathcal{Q}_{\preceq, k, D}} \mathbb{P}\left( \mathcal{W}_{2}\left(  \mu_n,\mathscr{P}_{\preceq \nu_m}^{cx}\right) \geq t_2(\alpha_{n,m}/2)  \right)\\
   & \quad + \sup_{\mathcal{Q}^c_{\preceq, k, D}(2t_2(\alpha_{n,m}/2))}\mathbb{P}\left(\mathcal{W}_{2}\left(\mu_{n},\mathscr{P}_{\preceq\nu_{m}}^{cx}\right)< t_2(\alpha_{n,m}/2)   \right) . 
\end{align*}
\end{enumerate}
\end{theorem}

\begin{remark}\label{rm:seperation-rate}
Theorem \ref{thm: consistency of the test} shows the consistency of \eqref{eqn:test-new-t}, i.e., our proposed test controls type I and II errors simultaneously and they become small as long as the number of samples grows and the significance level vanishes not too fast. Notice that this result relies on the choice of the separation rate $r_{n,m} = 2t(\alpha_{n,m})$. But this choice would not be optimal: if one can find smaller $r_{n,m}$ satisfying the consistency, then it would be better because such choice guarantees to regulate type I and II errors simultaneously over a larger family of distributions. For this reason, statisticians want to find a minimum separation rate. The \emph{minimax optimal separation distance} $\{r^*_{n,m} \}$ for a test (\cite[Definition 6.2.1]{gine2021mathematical}) is the minimum $r_{n,m}$ supporting a test to be consistent. Our choice $2t(\alpha_{n,m})$ automatically provides an upper bound of the minimax separation distance while its minimality is not confirmed. We leave the accomplishment of the minimax separation distance for future work.
\end{remark}

%


\section{Discussion on computation}\label{sec:computation}
A next natural question is how to compute the quantity $\mathcal{W}_{2}\left(  \mu_n,\mathscr{P}_{\preceq \nu_m}^{cx}\right)$. In this section we describe the set-up originally given in \cite[Equation (1.1)]{Alfonsi_2020}. In the next section we will discuss a numerical algorithm (See Section \ref{sec:entropic-Frank-Wolfe}) and experiments.

In \cite{Alfonsi_2020}, the authors provide the optimization problem to compute $\mathcal{W}_{2}\left(  \mu_n,\mathscr{P}_{\preceq \nu_m}^{cx}\right)$, stated as \cite[Equation (1.1)]{Alfonsi_2020}.
In \cite[Theorem 4.1]{Gozlan_Juillet20}, the authors consider a semi-discrete case in which $\nu_m$ is a discrete measure concentrated on the set of vertices of a simplex, and suggest how to compute $\mathcal{W}_{2}\left( \mu, \mathscr{P}_{\preceq \nu_m}^{cx} \right)$ for arbitrary $\mu$ with a compact support. The projection of $\mu$ onto $\mathscr{P}_{\preceq \nu_m}^{cx}$, denoted by $\overline{\mu}$, is the push-forward measure obtained by the composition of a translation and a projection map onto the convex hull of $\spt(\nu_m)$.

For the sake of completeness of the present paper, we provide our own proof of the computation of  $\mathcal{W}_{2}\left(  \mu_n,\mathscr{P}_{\preceq \nu_m}^{cx}\right)$ based on the idea of \cite[Theorem 4.1]{Gozlan_Juillet20}; in Lemma~\ref{lemma : barycentric characterization} and Theorem~\ref{thm : computing projection of mu_n} below. Although \cite{Alfonsi_2020} already resolves this question, it would be helpful for people in statistics or econometrics communities who might be not familiar with the Wasserstein projection or the martingale optimal transport to provide detailed explanation regarding the derivation of the formula to compute $\mathcal{W}_{2}\left(  \mu_n,\mathscr{P}_{\preceq \nu_m}^{cx}\right)$.

A key observation of the second proof of \cite[Theorem 4.1]{Gozlan_Juillet20} is that $\overline{\mu}$ should have the same barycenter of $\nu_m$, which is a simple consequence of convex order. For the very special case that $\mu_n$ and $\nu_m$ are supported in the set of vertices of a simplex this simple condition is sufficient, but obviously not enough for convex order in general cases; consider $\mu=\frac{1}{2}\delta_{e_1} + \frac{1}{2}\delta_{-e_1}$ and $\nu=\frac{1}{2}\delta_{e_2} + \frac{1}{2}\delta_{-e_2}$ for the standard basis vectors $e_1, e_2$ of $\R^2$, which have the same barycenter, but not in convex order. Therefore, it is insufficient for our case.

However, we can still use the idea of it. The hint is that for $\eta \in \mathscr{P}^{cx}_{\preceq \nu_m}$, every point of $\spt(\eta)$ should be a barycenter point of the subset of $\spt(\nu_m)$. We thus get the following lemma (proven in the \Cref{app: sec:computation}), which gives an effective parametrization of the (backward) convex order cone $\mathscr{P}^{cx}_{\preceq \nu_m}$ by the $m$-simplex.

\begin{lemma}\label{lemma : barycentric characterization}
For $m \in \mathbb{N}$, let $\Delta_m$ be the $m$-simplex. Assume that $\nu_m$ is a distribution over $m$-points, $\{y_1, \dots, y_m\}$. Let $T : \Delta_m \longrightarrow \text{conv}(\spt(\nu_m))$ such that $T(\boldsymbol{\alpha}) = \sum_{j=1}^m \alpha_j y_j$ for each $\boldsymbol{\alpha} := (\alpha_1, \dots, \alpha_m)^T \in \Delta_m$. Then,
\begin{equation}\label{eq : barycenter condition}
    \mathscr{P}^{cx}_{\preceq \nu_m} = \left\{ (T)_{\#}(\omega) : \omega \in \mathcal{P}(\Delta_m) \text{ s.t. }\int_{\Delta_m} \alpha_j d\omega(\boldsymbol{\alpha}) = \nu_m(y_j) \text{ for all $1 \leq j \leq m$} \right\}.
\end{equation}
\end{lemma}

The next theorem, proved in the \Cref{app: sec:computation}, gives the same result as \cite[Equation (1.1)]{Alfonsi_2020}, and describes how to use this parametrization to compute $\mathcal{W}_2(\mu_n , \mathscr{P}^{cx}_{\preceq \nu_m})$ and a (unique) projection of $\mu_n$ onto $\mathscr{P}^{cx}_{\preceq \nu_m}$. We do not claim novelty here.

\begin{theorem}\cite[Equation (1.1)]{Alfonsi_2020}\label{thm : computing projection of mu_n}
Suppose two empirical distributions
\[
    \mu_n = \sum_{i=1}^n \mu_n(x_i) \delta_{x_i}, \quad \nu_m = \sum_{j=1}^m \nu_m(y_j) \delta_{y_j}    
\]
are given. For each $x_i \in \spt(\mu_n)$ we write $\boldsymbol{\alpha}(x_i) := (\alpha_1(x_i), \dots, \alpha_m(x_i) )^T$. Consider the following constrained minimization problem for $\boldsymbol{\alpha}$: 

\begin{align}\label{eq: minimization for projection}
    & \min_{\{\boldsymbol{\alpha}(x_1), \dots, \boldsymbol{\alpha}(x_n) \}} \sum_{i=1}^n \mu_n(x_i) \left( \sum_{j=1}^m \alpha_j(x_i) y_j - x_i \right)^2  \\ \nonumber
    & \text{s.t. } \boldsymbol{\alpha}(x_i) \in \Delta_m \text{ for all $x_i \in \spt(\mu_n)$}, \quad \sum_{i=1}^n \alpha_j(x_i) \mu_n(x_i) = \nu_m(y_j) \text{ for all $y_j \in \spt(\nu_m)$}. 
\end{align}
There exists a unique minimizer $\{ \boldsymbol{\alpha}^*(x_1), \dots, \boldsymbol{\alpha}^*(x_n) \}$ of \eqref{eq: minimization for projection} which induces the projection $\overline{\mu}_n$  of $\mu_n$ onto $\mathscr{P}^{cx}_{\preceq \nu_m}$ given in this way: for any measurable $E \subseteq \text{conv}(\spt(\nu_m))$, 
\begin{equation}\label{eq: form of projection}
    \overline{\mu}_n(E) = \sum_{i=1}^n \mu_n(x_i) \mathds{1}_{E} \left( \sum_{j=1}^m \alpha^*_j(x_i) y_j \right).
\end{equation}
\end{theorem}

One can write \eqref{eq: minimization for projection} in terms of matrices. Let $x_i = (x_{i1}, \cdots, x_{id})$, $y_j = (y_{j1}, \cdots, y_{jd})$ and
\begin{align*}
    &\mathbf{A} =
    \left[ \begin{array}{ccc}
    \alpha_1(x_1) & \ldots & \alpha_m(x_1) \\
    \vdots & \ddots & \vdots \\
    \alpha_1(x_n) & \ldots & \alpha_m(x_m)
    \end{array} \right], \quad 
    \mathbf{Y} =
    \left[ \begin{array}{ccc}
    y_{11} & \ldots & y_{1d} \\
    \vdots & \ddots & \vdots \\
    y_{m1} & \ldots & y_{md}
    \end{array} \right], \quad
    \mathbf{X} =
    \left[ \begin{array}{ccc}
    x_{11} & \ldots & x_{1d} \\
    \vdots & \ddots & \vdots \\
    x_{n1} & \ldots & x_{nd}
    \end{array} \right]\\
    &\boldsymbol{\mu_n}=\left[ \begin{array}{ccc}
    \mu_n(x_1) \\
    \vdots \\
    \mu_n(x_n)
    \end{array} \right], \quad    
    \boldsymbol{\nu_m}=\left[ \begin{array}{ccc}
    \nu_m(y_1) \\
    \vdots \\
    \nu_m(y_m)
    \end{array} \right], \quad
    \boldsymbol{1_k}=\left[ \begin{array}{ccc}
    1 \\
    \vdots \\
    1
    \end{array} \right].
\end{align*}
Then, \eqref{eq: minimization for projection} can be written as
\begin{equation}
\begin{aligned}\label{eq: minimization for projection matrix form}
    &\min_{\mathbf{A}} \text{trace}\left( (\mathbf{A}\mathbf{Y} - \mathbf{X})^T \text{diag}(\boldsymbol{\mu_n}) (\mathbf{A}\mathbf{Y} - \mathbf{X}) \right)\\
    &\text{s.t. } \mathbf{A}^T \boldsymbol{\mu_n} = \boldsymbol{\nu_n}, \quad \mathbf{A}\boldsymbol{1_m} = \boldsymbol{1_n}, \quad \mathbf{A} \geq 0. 
\end{aligned}    
\end{equation}
Here, $\text{diag}(\boldsymbol{\mu_n})$ denotes the diagonal matrix whose diagonal is $\boldsymbol{\mu_n}$ and $\mathbf{A} \geq 0$ means entrywise nonnegativity. Notice that this is a quadratic semidefinite program, hence convex. Since $\{ \mathbf{A}Y \  : \  \mathbf{A} \in \mathbb{R}^{m \times n}, \,  \mathbf{A}\boldsymbol{1_m} = \boldsymbol{1_n}, \mathbf{A} \geq 0 \}$ is indeed a convex hull of $\spt(\nu_m)$, one can regard this problem as a regression problem restricted to a convex hull. A folklore in computational geometry tells that a projection map onto a convex hull has no closed form in general, hence it is impossible to obtain the closed form of a solution of \eqref{eq: minimization for projection matrix form}. Furthermore, a convex projection in our case is more complicated than the usual one. In classical cases, the image of a convex projection always lies in the surface of a given convex hull: however in our case, due to the constraint of marginal measures, it maps some points into the inside of the convex hull.

\section{Algorithms and experiments}\label{sec:experiment}
In this section, a possibly efficient algorithm is proposed. The effectiveness of our approach is also empirically verified.

Let us compare the numerical experiments below with those considered in \cite{Alfonsi_2020}. There are three numerical examples in \cite{Alfonsi_2020}. The first two examples are about the Wasserstein projection distance in one or two dimensions, which are also relevant to our problem, and the last one is about the optimization over couplings with the marginal and the martingale constraints, which is of interest to finance community.

Focusing on the first two examples of
\cite{Alfonsi_2020}, in the first example, the Wasserstein projected distances between $\mu_n$ and several objects are concerned. Those objects include not only the projection $\mu_n$ onto $\mathscr{P}_{\preceq \nu_n}^{cx}$ but also the minimum of $\mu_n$ and $\nu_m$ with respect to convex order. Notice that the minimum (the maximum) of for any $\mu$ and $\nu$ with respect to convex order in one dimension is well defined; in fact $\mathcal{P}(\mathbb{R})$ with convex order exhibits the complete lattice structure (see \cite{Kertz_Rosler2000}). The second example shows the structure of martingale coupling in two dimension. It should be noted again that all the above examples assume $\mu \preceq \nu$ a priori, so do not include the case of $\mu \not \preceq \nu$.

In our numerical examples, we consider not only the case of $\boldsymbol{H}_0 : \mu \preceq \nu$ but also that of $\boldsymbol{H}_A : \mu \not \preceq \nu$. For statistical applications, the proposed statistic $\mathcal{W}_2\left( \mu_n, \mathscr{P}_{\preceq \nu_n}^{cx} \right)$ should perform differently under those different scenarios. In this sense, it should converge to $0$ fast as $n, m$ grow under $\boldsymbol{H}_0$, and stay away from $0$ sufficiently enough under $\boldsymbol{H}_A$.

First, we give a simple example that demonstrates how $\mathcal{W}_2\left( \mu_n, \mathscr{P}_{\preceq \nu_n}^{cx} \right)$ behaves as $n, m$ grow under the null hypothesis.

\begin{example}[{See \cite[Examples 5.2.1 and 5.2.2]{Alfonsi_2020}}]
On $\mathbb{R}^2$, let $N(0, I_2)$ be the standard gaussian distribution. Consider $\mu = \text{Unif }[0,1]^2$, the uniform distribution over $[0,1]^2$, and $\nu = \text{Unif }[0,1]^2 * N(0, I_2)$, the convolution of $\text{Unif }[0,1]^2$ and $N(0, I_2)$. 
It is straightforward that $\mu \preceq \nu$ since for $X \sim \mu$ and $Y \sim \nu$, the martingale condition $\mathbb{E}_{\pi}[Y | X] = X$ is achieved by $\pi(y|x) = N(x, I_2)$, the isotropic gaussian with mean $x$. Let $\mu_n$ and $\nu_n$ be the empirical distributions of $\mu$ and $\nu$ with independent $n$-samples, respectively. Also, both have the uniform weight over samples. In this experiment, we use CVXPY python code developed in \cite{diamond2016cvxpy, agrawal2018rewriting}.

Combining \Cref{thm: consistency of projected Wasserstein-1} with \Cref{thm : upper bound of expected distance}, it holds that with high probability
\[
    \mathcal{W}_2\left( \mu_n, \mathscr{P}_{\preceq \nu_n}^{cx} \right) \leq O\left( n^{-\frac{1}{4}} \right).
\]
\Cref{projection} shows a numerical result which matches the theoretical expectation pretty well. The $x$-axis represents the sample size $n$, and for each $n$, we draw samples 5 times independently and compute their $\mathcal{W}_2\left( \mu_n, \mathscr{P}_{\preceq \nu_n}^{cx} \right)$.  
As we see $ \mathcal{W}_2\left( \mu_n, \mathscr{P}_{\preceq \nu_n}^{cx} \right)$ tends to decrease as $n$ increases though there is some fluctuation. Due to the concentration phenomenon, the fluctuation is expected to vanish as $n$ becomes larger.

\Cref{uniform+gaussian} shows the geometry of the Wasserstein projection when the sample size $n= 100$.
Red and blue dots represent $\mu_n$ and $\nu_n$, respectively, and the set of green dots denotes the projection $\overline{\mu}_n$ of $\mu_n$ onto $\mathscr{P}_{\preceq \nu_n}^{cx}$. Notice that since $\mu \preceq \nu$, $\mathcal{W}_2\left( \mu_n, \mathscr{P}_{\preceq \nu_n}^{cx} \right) \to 0$ as $n \to \infty$. Pictorially, as $n$ grows, one expects that red dots and  green ones become closer. The right picture of \Cref{uniform+gaussian} describes the discrepancy between $\mu_n$ and $\overline{\mu}_n$. 
By \Cref{thm : computing projection of mu_n}, we know that there is a $1$-Lipschitz map $\Phi_n$, the optimal transport map, such that $\overline{\mu}_n = (\Phi_n)_{\#}(\mu_n)$. Since $\mathcal{W}_2\left( \mu_n, \mathscr{P}_{\preceq \nu_n}^{cx} \right) \to 0$ as $n \to \infty$, $\Phi_n$
converges to the identity map as $n \to \infty$.
\begin{figure}[h]
		\centering    \includegraphics[width=0.5\linewidth]{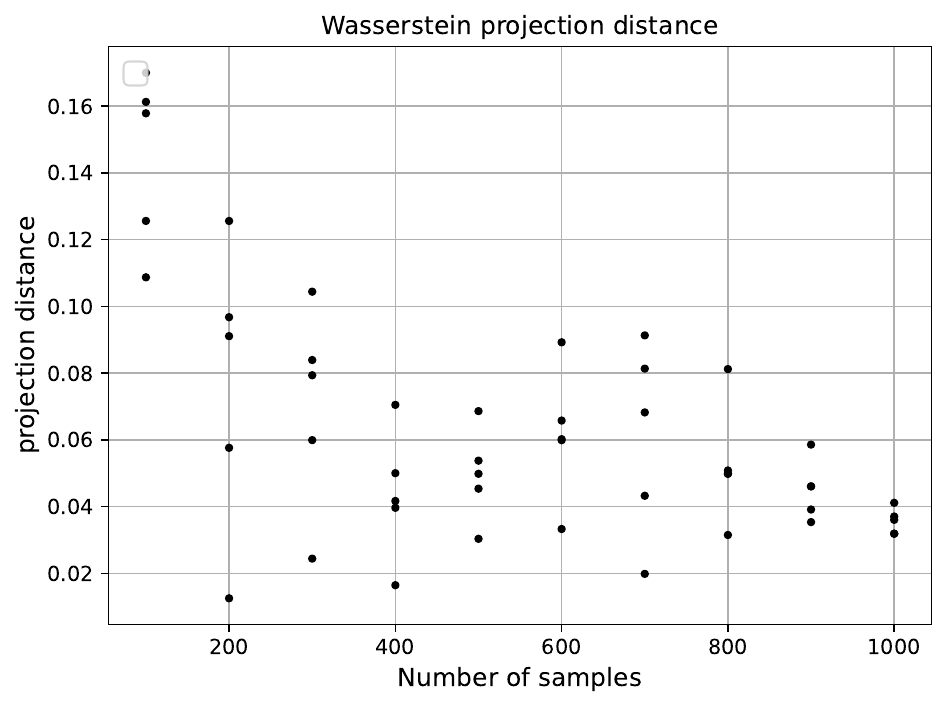}
		\centering    
		\caption{The plot of $\mathcal{W}_2\left( \mu_n, \mathscr{P}_{\preceq \nu_n}^{cx} \right)$.}
		\label{projection}
	\end{figure}
\begin{figure}[h]
		\centering    \includegraphics[width=0.8\linewidth]{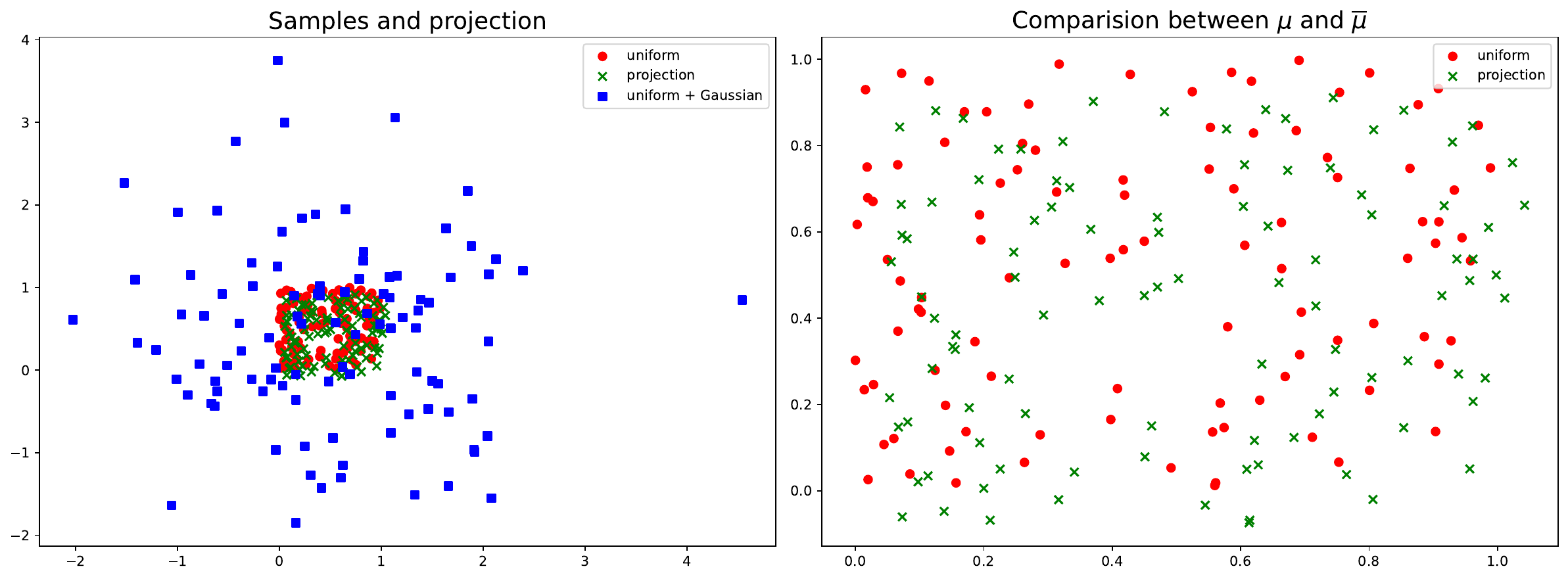}
		\centering    
		\caption{The geometry of the Wasserstein projection.}
		\label{uniform+gaussian}
	\end{figure}
\end{example}

In fact, CVXPY is not efficient for a large scale problem with high dimensional data. To pursue more efficient algorithms, we introduce gradient descent based scheme for the optimization problem \eqref{eq: minimization for projection matrix form}, equivalently, \eqref{eq: minimization for projection}, which is derived from the squared Wasserstein backward projection $\mathcal{W}^2_{2}\left(
\mu,\mathscr{P}_{\preceq\nu}^{\text{cx}}\right)$.

For simplicity we sample equal number of points from both marginals. Consider
$\mu_{n}=\frac{1}{n}\sum_{k=1}^{n}\delta_{X_{i}}$, $\nu_m=\frac{1}{n}\sum_{k=1}^{n}\delta_{Y_{i}}$ where $X_{1},...,X_{n}$ and $Y_{1},...,Y_{m}\in\mathbb{R}^{d}$ are i.i.d. samples drawn from $\mu$ and $\nu$, respectively. Then the optimization problem \eqref{eq: minimization for projection matrix form} is simplified to
\begin{equation}\label{eq:BkWd_opt}
    \min_{\substack{\pi\in\mathbb{R}^{n}\times\mathbb{R}^{n}\\\pi\mathbf{1=1,}\pi^{T}\mathbf{1=1}}}\mathcal{J}\left(  \pi\right)  \triangleq\frac{1}{n}\left\Vert \pi\mathbf{Y}-\mathbf{X}\right\Vert _{\text{F}}^{2},
\end{equation}
where $\mathbf{X}$ (resp. $\mathbf{Y}$) is the matrix formed by vertically
stacking $X_{1}^{T},...,X_{n}^{T}$ (resp. $Y_{1}^{T},...,Y_{m}^{T}$), 
$\left\Vert \cdot\right\Vert _{\text{F}}$ is the Frobenius norm, and $\mathbf{1}\in\mathbb{R}^{d}$ is the vector of all ones.

It should be noted that the number of variables and constraints of the optimization \eqref{eq:BkWd_opt} are both independent of the
space dimension, $d$. Increasing the dimension does not change the underlying search space, although it could affect the individual convergence performance.

The Frank-Wolfe algorithm \cite{frank1956algorithm} is well-known to solve the optimization \eqref{eq:BkWd_opt}. It is an iterative gradient method for constrained optimization. At each iterative step $\pi_{k}\in\mathbb{R}^{n}\times\mathbb{R}^{n}$, $k \geq 0$, the algorithm predicts the movement
of the next step by solving a linear programming problem, called \emph{oracle}. The iteration steps $\pi_{0},\pi_{1}, \dots ,\pi_{k}, \dots$ actually produce a \emph{gradient flow} along the backward convex order cone $\mathscr{P}_{\preceq\nu_m}^{\text{cx}}$. Below we discuss a regularized Frank-Wolfe algorithm.

\subsection{The entropic Frank-Wolfe algorithm}\label{sec:entropic-Frank-Wolfe}
It is worth mentioning that despite the name, the algorithm does not require regularization of the objective function, but instead uses regularization in the process of optimization, thus falls under different category from the classical entropic optimal transport method and, to clear any possible confusion, the regularized martingale optimal transport method proposed in \cite{Alfonsi_2020}.

Now we present a particular version of the Frank-Wolfe algorithm which equips it with an \emph{entropic optimal transport oracle}. We call it entropic Frank-Wolfe algorithm. It turns out to be useful for the
optimization we have at hand. At each iteration step $\pi_{k}\in\mathbb{R}^{n}\times\mathbb{R}^{n},$ it solves the linear programming problem whose
optimal solution is used to compute the movement of the next step:
\begin{equation}
\min\limits_{\pi\in\mathcal{E}}\nabla\mathcal{J}\left(  \pi_{k}\right)
\odot\pi, \label{eq:BkWd_opt_lin}%
\end{equation}
where $\odot$ is the Hadamard product (the element-wise product) of matrices,
\[
\nabla\mathcal{J}\left(  \pi_{k}\right)  =\left.  \nabla_{\pi}\left(  \frac
{1}{n}\left\Vert \pi\mathbf{Y}-\mathbf{X}\right\Vert _{\text{F}}^{2}\right)
\right\vert _{\pi_{k}}=\frac{2}{n}\left(  \pi_{k}\mathbf{Y}-\mathbf{X}\right)
\mathbf{Y}^{T}, 
\]
and%
\[
\mathcal{E}=\left\{  \pi\in\mathbb{R}^{n}\times\mathbb{R}^{n}:\pi
\mathbf{1=1,}\pi^{T}\mathbf{1=1}\right\}  .
\]
Various computational methods are available for solving the
linear programming problem \eqref{eq:BkWd_opt_lin}. For
example, the simplex method, the interior point method or the multi-scale linear
programming \cite{oberman2020solution}. An optimal solution of $\left(
\ref{eq:BkWd_opt_lin}\right)  $ would point us to the direction of the next
gradient descent movement. Since an accurate moving direction is not
necessary, we can thus replace \eqref{eq:BkWd_opt_lin} with
an optimization problem which produces an optimal solution close to that of \eqref{eq:BkWd_opt_lin}. Specifically, we fix $\varepsilon
_{k}>0$ and consider%
\begin{equation}
\min\limits_{\gamma\in\mathcal{E}}\left\{  \nabla\mathcal{J}\left(
\pi_{k}\right)  \odot\gamma+\varepsilon_{k}\text{KL}\left(  \gamma\left\vert \mathbf{1}  \mathbf{\otimes}
\mathbf{1}  \right.  \right)  \right\}  .\label{eq:BkWd_opt_reg}%
\end{equation}
KL$\left(  \cdot|\mathbf{\cdot}\right)  $ denotes the sum of element-wise
Kullback-Leibler divergences, and  $\varepsilon_{k}$ satisfies $\varepsilon
_{k}\rightarrow0$ as $k\rightarrow\infty.$ 
A benefit of
solving \eqref{eq:BkWd_opt_reg} over \eqref{eq:BkWd_opt_lin} is that, due to the structure of the
constraints, it can be regarded as an entropic optimal transport with $\gamma$
being the probability coupling between $\frac{1}{n}\mathbf{1}$ and $\frac
{1}{n}\mathbf{1}$ (with a normalization). Many light-weight algorithms are available for entropic
optimal transport, e.g., the Sinkhorn algorithms or its variants
\cite{peyre2019computational}, \cite{altschuler2017near}. 
We call this Frank-Wolfe algorithm equipped with the oracle \eqref{eq:BkWd_opt_reg} the \emph{entropic Frank-Wolfe algorithm}.

\subsection{Convergence}
Convergence of the Frank-Wolfe algorithm in a general constrained optimization problem is well-known \cite{jaggi2013revisiting}. Fast convergence rate can be expected for our problem \eqref{eq:BkWd_opt} that has a certain special structure.  
Note the objective function is convex, although not uniformly. Define
\[
\mathcal{E}_{atom}=\left\{  E_{ij}\in\mathbb{R}^{n}\times\mathbb{R}%
^{n}:i,j=1,...,n\right\}  .
\]
where $E_{ij}\in\mathbb{R}^{n}\times\mathbb{R}^{n}$ has zero elements except the value $1$ at the $(i,j)$-th slot. Then the optimization is actually performed over the
convex hull of $\mathcal{E}_{atom}$, i.e., $\mathcal{E=}$ $\text{conv}\left(\mathcal{E}_{atom}\right)$. This enables us to approximate the movement of the next step by those from $\mathcal{E}_{atom}$; precisely, the oracle only has to optimize over $\mathcal{E}_{atom}$ rather than $\mathcal{E}$. Simplex
method and its variants are suitable for the task as an oracle over $\mathcal{E}_{atom}$. In this setting, (approximate) linear convergence is achievable \cite[the general convex case of Theorem 1]{lacoste2015global}. One obstacle to linear convergence is indeed the convexity of the problem at hand, and the usual regularization trick can be employed to fix the issue. For
other oracles, such as the \eqref{eq:BkWd_opt_reg}, the convergence rate has not been established, its study is out of the scope of the article and we will pursue this elsewhere.

The next example demonstrates the efficiency of the entropic Frank-Wolfe algorithm and the ability of the Wasseerstein projection as a way of sampling measures in convex order. Here $\boldsymbol{H_A}: \mu \npreceq \nu$ is taken into account.

\begin{example}
Consider two dimensional distributions%
\[
\mu\sim N\left(  \left(
\begin{array}
[c]{c}%
0\\
0
\end{array}
\right)  ,\left(
\begin{array}
[c]{cc}%
2 & -2\\
-2 & 3
\end{array}
\right)  \right)  ,\text{ }\nu\sim N\left(  \left(
\begin{array}
[c]{c}%
1\\
1
\end{array}
\right)  ,\left(
\begin{array}
[c]{cc}%
3 & -2\\
-2 & 4
\end{array}
\right)  \right) 
\]
where $\nu$ is obtained by diffusing $\mu$ and shifting by the vector $\left(
1,1\right)  .$ Clearly $\mu$ is not in convex order with $\nu$.
Now $10^{4}$ samples are drawn respectively from $\mu$ and $\nu.$ The
empirical measures are written as $\mu_{n}$ and $\nu_m$. We start the
entropic Frank-Wolfe algorithm \eqref{eq:BkWd_opt_reg} with the initial $\pi_{0}$ being a $10^{4}%
\times10^{4}$ matrix with each row equal to%
\[
\left(  10^{-4},...,10^{-4}\right)  \in\mathbb{R}^{10^{4}}.
\]
This corresponds to a Dirac measure on the cone $\mathscr{P}_{\preceq\nu_m%
}^{\text{cx}}$ concentrating on the barycenter of $\nu_m$. In around $10$
iterations, the algorithm reaches a minimal objective value close to $2$ and stabilizes
in further iterations, which indicates $\mu_{n}$ is not in the cone
$\mathscr{P}_{\preceq\nu_m}^{\text{cx}}$. \Cref{fig:gradflow} shows
some snapshots of the evolution of the probability along the cone.
\end{example}

\begin{figure}[h]
\centering
\includegraphics[width=\textwidth]{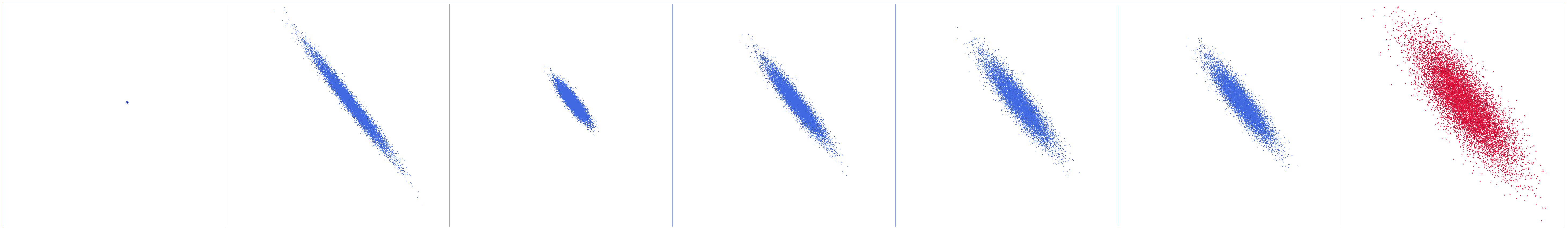} \caption{Left-most:
the initial distribution concentrating on the barycenter of $\nu_m.$
Right-most: the empirical distribution $\nu_m.$ Middle: slices of the gradient flow along the convex order cone $\mathscr{P}_{\preceq\nu_m}^{\text{cx}}$ (from left to right) generated by the entropic Frank-Wolfe algorithm, that is expected to converge to $\overline{\mu}_n$ the projection to $\mathscr{P}_{\preceq\nu_m}^{\text{cx}}$. Each square is $\left[  -9,9\right] \times \left[-9,9\right] .$ The shape of the distribution quickly becomes stable in a few iterations.}%
\label{fig:gradflow}%
\end{figure}

\section{Conclusions and future directions}\label{sec: conclusion}
In this paper, we have discussed the hypothesis testing problem for convex order of probability measures. We have established a new methodology for this problem, using the Wasserstein projection over the Wasserstein space. In full generality, we demonstrated that the Wasserstein projected distance is stable, hence its empirical version is consistent with the true one. We further obtained the rate of convergence of the projected Wasserstein distance of empirical measures under either the log-Sobolev inequality or bounded support assumptions.
In addition to providing theoretical understanding, we also have proposed  a numerical scheme for the projection distance and provide some experiments based on this scheme. In what follows we continue to expand this discussion, while at the same time provide a few perspectives on future work.

First, it is of interest to analyze other stochastic orders and connect them with optimal transport view points. Often variants of convex order  are more frequently utilized in economics and risk management. For example, one can choose a different defining class $\mathcal{A}'$ which is the set of increasing convex functions, increasing with respect to lexicographic order $\leq$ on $\mathbb{R}^d$, to construct the increasing convex order. Another example is the subharmonic order for probability measures over a convex domain, which is induced by another defining class, namely, the set of all bounded continuous subharmonic functions. These are two typical other stochastic orders that are considered in various fields. It is valuable to answer similar questions on those orders.

Second, it would be of interest to sharpen the theory we developed here. For example, most constants appearing in our estimates are loose, which worsens convergence rate results. Also, the rate of convergence result without the log-Sobolev inequality or bounded supports assumptions is of course desirable. An appropriate next in this regard is to provide a central limit theorem for our estimator, which advances $p$-value and type I/II errors of it.

Third, it may be possible to apply entropic optimal transport methods in our setting. In spite of the curse of dimensionality of the rate of convergence in `vanilla' Wasserstein distances, the entropic optimal transport problems enjoy $n^{-\frac{1}{2}}$ rate of convergence. Furthermore, the central limit theorem holds for entropic optimal transport: see \cite{MR4041704, gonzalezsanz2022weak, gonzálezsanz2023weak, 10.1214/23-AAP1986}. 
From both theoretical and practical points of view, the wide computational utility of entropic optimal transport tools is enticing; yet, how exactly to apply this toolkit in our setting remains non-obvious.

Finally, it is of interest to understand computational difficulty precisely. At first glance \eqref{eq: minimization for projection} is nothing but the usual convex projection problem, but it is slightly different from it. Usual convex projection always has the property  that any projected point lies in the surface of the target convex hull. On the contrary, due to the marginal constraints in our setting, the relevant projection can reside in the interior of the convex hull. Does this imply a big difference? Furthermore, to the best our knowledge, the entropic Frank-Wolfe algorithm is introduced first in this paper. It employs a gradient descent on the convex order cone, which we want to understand further.

%
%
%
%

\begin{acks}[Acknowledgments]
 We would like to express our gratitude to Nathael Gozlan, Philippe Chone and Francis Kramarz for drawing our attention to the applications of the Wasserstein projection in economics. Additionally, we extend our gratitude to Aur\'elien Alfonsi and Benjamin Jourdain for highlighting the connection to their research \cite{Alfonsi_2020} as well as \cite{Jourdain_2023}. These references were not included in the earlier version of this paper. We would like to thank Chang-Jin Kim and Yanqin Fan for valuable comments about stochastic order in economics context and recommendation of related previous literature, and Sharvaj Kubal for introducing us to CVXPY. Lastly, we extend our sincere appreciation to Xiaohui Chen and Rami Tabri for their valuable suggestions on the manner of stating our hypotheses and for related references about nonparametric hypothesis testing; such a fruitful discussion was possible thanks to the Summer School on Optimal Transport, Stochastic Analysis and Applications to Machine Learning, held in 2024 at Korea Advanced Institute of Science and Technology (KAIST), jointly organized by Stochastic Analysis and Application Research Center (SAARC) at KAIST and Pacific Institute for the Mathematical Sciences (PIMS).
\end{acks}

\begin{funding}
JK thanks to the PIMS postdoctoral fellowship  through the Kantorovich Initiative PIMS Research Network (PRN) as well as National Science Foundation grant NSF-DMS 2133244. YHK is partially supported  by the Natural Sciences and Engineering Research Council of Canada (NSERC), with Discovery Grant RGPIN-2019-03926. YHK and AW are partially supported by Exploration Grant NFRFE-2019-00944 from the New Frontiers in Research Fund (NFRF). AW acknowledges the support of the Burroughs Wellcome Fund grant 2019-014832. YHK is a member of the Kantorovich Initiative (KI) that is supported by PIMS Research Network (PRN) program. We thank PIMS for their generous support; report identifier PIMS-20240605-PRN01. Part of this work is completed during YHK’s visit at KAIST,  through the support from SAARC, and YR's visit at the University of British Columbia (UBC). YHK and YR thank KAIST's and UBC's hospitality and excellent environment, respectively.
\end{funding}

\bibliographystyle{imsart-number} 
\bibliography{references.bib}

\begin{appendix}
%

\section{Proofs of section~\ref{sec : Convergence rate}}\label{app: sec : Convergence rate}

\begin{proof}[Proof of \Cref{lem: concentration}]
\begin{enumerate}
    \item[(i)] Consider the random function
\begin{equation}\label{eq:Wproj_func}
    W\left(  X_{1},...,X_{n},Y_{1},...,Y_{m}\right)  =\mathcal{W}_{2}\left(\mu_n,\mathscr{P}_{\preceq \nu_m}^{\text{cx}}\right).
\end{equation}
Due to the the Lipschtiz stability of the Wasserstein projection,
\begin{align*}
    \left\vert W\left(  X_{1},...,X_{i}^{\prime},...,X_{n},Y_{1},...,Y_{m}\right) -W\left(  X_{1},...,X_{i},...,X_{n},Y_{1},...,Y_{m}\right)  \right\vert &\leq \frac{1}{m^{1/2}}\left\vert X_{i}^{\prime}-X_{i}\right\vert,\\
    \left\vert W\left(  X_{1},...,X_{n},Y_{1},...,Y_{j}^{\prime},...,Y_{m}\right) -W \left( X_{1},...,X_{n},Y_{1},...,Y_{j},...,Y_{m}\right)  \right\vert &\leq \frac{1}{n^{1/2}}\left\vert Y_{j}^{\prime}-Y_{j}\right\vert .
\end{align*}
Now the conclusion follows from the product measure reduction property of log-Sobolev inequality and the classical concentration inequality argument
\cite[Section 2.3]{ledoux1999concentration}.
    \item[(ii)] Since $\mu$ and $\nu$ have bounded supports, we can assume without loss of generality that random samples
$X_{1},...,X_{n}\sim\mu,$ $Y_{1},...,Y_{m}\sim\nu$ are all contained in a bounded domain $\Omega$. Recall the duality of the Wasserstein projection $\left(  \ref{eq:Wproj_func}%
\right)  ,$%
\[
    W^2\left(  X_{1},...,X_{n},Y_{1},...,Y_{m}\right)  =\sup_{\varphi \in\mathcal{A}_{cx}\cap C_{b,2}\left(  \Omega\right)  }\left\{  \int Q_{2}\left(  \varphi\right)  d\mu_{n}-\int\varphi d\nu_{m}\right\}  ,
\]
where
\[
    Q_{2}\left(  \varphi\right)  \left(  x\right)  =\inf_{y\in\Omega}\left\{ \varphi\left(  y\right)  +\left\vert x-y\right\vert ^{2}\right\} .
\]
Following the idea of \cite[Theorem 10.1]{yh_yl_stochastic_order} or, in view of the relation between weak optimal transport and backward projection
reviewed in the preliminaries, \cite[Theorem 6.1]{Gozlan_Juillet20}, we see that the optimal dual potential for $W^2$ is attained by some $\varphi_{0}\in\mathcal{A}_{cx}\cap C_{b,2}$ with Lipschitz constant $2D$. By adding a constant to all test functions in $\mathcal{A}_{cx}\cap
C_{b,2}$, we are not changing the value of the supremum. Thus we can assume that all test function $\varphi$ satisfies that
\[
    \inf_{y\in\Omega}\varphi\left(  y\right)  =0.
\]
Therefore,
\[
    W^2=\sup_{\substack{\varphi\in \mathcal{A}_{cx}\cap C_{b,2}\left(  \Omega\right)  \\0\leq\varphi \leq2D^{2}}}\left\{  \frac{1}{n}\sum_{i=1}^{n}Q_{2}\left(  \varphi \right)  \left(  X_{i}\right)  -\frac{1}{m}\sum_{j=1}^{m}\varphi\left(Y_{j}\right)  \right\}  .
\]
Note that for any $0\leq\varphi\leq2D^{2}$ and $0\leq Q_{2}\left(  \varphi\right) \leq \varphi \leq 2D^{2}$. Hence
\begin{align*}
    &  W^2\left(  X_{1},...,X_{n}^{\prime},Y_{1},...,Y_{m}\right)  -W^2\left( X_{1},...,X_{n},Y_{1},...,Y_{m}\right)\\   
    & \leq\sup_{\substack{\varphi^{\prime}\in\mathcal{A}_{cx}\cap C_{b,2}\left(  \Omega\right)  \\0\leq\varphi^{\prime}\leq2D^{2}}}\left\{  \frac{1}{n}\left[  \sum_{i=1}^{n-1}Q_{2}\left(  \varphi^{\prime}\right)  \left(  X_{i}\right)  +Q_{2}\left(  \varphi^{\prime}\right)  \left(X_{n}^{\prime}\right)  \right]  -\frac{1}{m}\sum_{j=1}^{m}\varphi^{\prime }\left(  Y_{j}\right)  \right\}\\
    &\quad - \sup_{\substack{\varphi\in\mathcal{A}_{cx}\cap C_{b,2}\left( \Omega\right)  \\0\leq\varphi\leq2D^{2}}} 
   {  \left\{ \frac{1}{n} \left[\sum_{i=1}^{n}Q_{2}\left(  \varphi\right)  \left(X_{i}\right)  \right] -\frac{1}{m}\sum_{j=1}^{m}\varphi\left(  Y_{j}\right) \right\}
    }
    \\
    &    
\leq\sup_{\substack{\varphi^{\prime}\in\mathcal{A}_{cx}\cap C_{b,2}\left(  \Omega\right)  \\0\leq\varphi^{\prime}\leq2D^{2}}}
\Big\{   \frac{1}{n}\left[  \sum_{i=1}^{n-1}Q_{2}\left(  \varphi^{\prime}\right)  \left(  X_{i}\right)  +Q_{2}\left( \varphi^{\prime}\right)  \left(X_{n}^{\prime}\right)  \right]  -\frac{1}{m}\sum_{j=1}^{m}\varphi^{\prime}\left(  Y_{j}\right)
    \\
    &  \quad \qquad\qquad\qquad
     -\left[  \frac{1}{n}\sum_{i=1}^{n}Q_{2}\left(\varphi^{\prime}\right)  \left(  X_{i}\right)  -\frac{1}{m}\sum_{j=1}^{m}\varphi^{\prime}\left( Y_{j}\right)  \right]  
      \Big\}
      \\
    \\
    &  =\frac{1}{n}\sup_{\substack{\varphi^{\prime}\in\mathcal{A}_{cx}\cap C_{b,2}\left(  \Omega\right) \\0\leq\varphi^{\prime}\leq2D^{2}}}\left\{  Q_{2}\left(  \varphi^{\prime}\right)  \left(  X_{n}^{\prime }\right)  -Q_{2}\left( \varphi^{\prime}\right)  \left(  X_{n}\right)
\right\}  \leq\frac{2D^{2}}{n}\leq\frac{2D^{2}}{m\wedge n}.
\end{align*}
Similarly 
\[
    W^2\left(  X_{1},...,X_{n},Y_{1},...,Y_{m}^{\prime}\right)  -W^2\left( X_{1},...,X_{n},Y_{1},...,Y_{m}\right)  \leq\frac{2D^{2}}{m}\leq\frac{2D^{2}}{m\wedge n}.
\]
Applying McDiarmid's inequality, it follows that
\begin{equation}\label{eq: squared concentrate_bd_supp}
    \mathbb{P}\left(  \left\vert \mathcal{W}_{2}^{2}\left(  \mu_{n},\mathscr{P}_{\preceq \nu_{m}}^{\text{cx}}\right)  -\mathbb{E}\left[ \mathcal{W}_{2}^{2}\left(  \mu_{n},\mathscr{P}_{\preceq\nu_{m}}^{\text{cx}}\right)  \right] \right \vert \geq t\right)  \leq 2\exp\left(  -\frac{\left(  m\wedge n\right)  ^{2}t^{2}}{2\left(  m+n\right) D^{4}}\right)  .
\end{equation}

Using the inequality
\[
    \left\vert \sqrt{a}-\sqrt{b}\right\vert \leq \sqrt{\left\vert a-b\right\vert },\text{ for all } a,b\geq 0,
\]
we obtain that for any $t\geq0,$%
\begin{align*}
    &  \mathbb{P}\left(  \left\vert \mathcal{W}_{2}\left(  \mu_{n},\mathscr{P}_{\preceq\nu_{m}}^{\text{cx}}\right)  -\mathcal{W}_{2}\left(  \mu,\mathscr{P}_{\preceq\nu}^{\text{cx}}\right)  \right\vert \geq t\right)\\
    &  \leq \mathbb{P}\left(  \left[  \left\vert \mathcal{W}_{2}^{2}\left( \mu_{n},\mathscr{P}_{\preceq\nu_{m}}^{\text{cx}}\right)  -\mathcal{W}_{2}^{2}\left(  \mu,\mathscr{P}_{\preceq\nu}^{\text{cx}}\right)  \right\vert \right]  ^{1/2} \geq t\right)\\
    &=\mathbb{P}\left(  \left\vert \mathcal{W}_{2}^{2}\left(  \mu_{n},\mathscr{P}_{\preceq\nu_{m}}^{\text{cx}}\right)  -\mathcal{W}_{2}^{2}\left( \mu,\mathscr{P}_{\preceq\nu}^{\text{cx}}\right)  \right\vert \geq t^{2}\right)  .
\end{align*}
Combining the above with \eqref{eq: squared concentrate_bd_supp}, the conclusion follows.
\end{enumerate}
\end{proof}

\begin{proof}[Proof of \Cref{prop: tail_probability}]
For any $t\geq 0$,
\begin{align*}
    & \mathbb{P}\left(  \left\vert \mathcal{W}_{2}\left(  \mu_{n},\mathscr{P}_{\preceq \nu_{m}}^{\text{cx}}\right)  -\mathcal{W}_{2}\left(  \mu,\mathscr{P}_{\preceq\nu}^{\text{cx}}\right)  \right\vert \geq t\right)\\
    &  \leq \mathbb{P}\left(  \left\vert \mathcal{W}_{2} \left(  \mu_{n},\mathscr{P}_{\preceq\nu_{m}}^{\text{cx}}\right)  -\mathbb{E}\left[\mathcal{W}_{2} \left( \mu_{n},\mathscr{P}_{\preceq\nu_{m}}^{\text{cx}}\right)  \right]  \right\vert \geq\frac{t}{2}\right)\\
    &\quad +\mathbb{P}\left(  \mathbb{E} \left\vert \mathcal{W}_{2} \left( \mu,\mathscr{P}_{\preceq\nu}^{\text{cx}}\right)  -\left[  \mathcal{W}_{2} \left(  \mu_{n},\mathscr{P}_{\preceq\nu_{m}}^{\text{cx}}\right)  \right]  \right\vert \geq\frac{t}{2}\right)  \\
    & \leq \mathbb{P}\left(  \left\vert \mathcal{W}_{2}\left(  \mu_{n},\mathscr{P}_{\preceq\nu_{m}}^{\text{cx}}\right)  -\mathbb{E}\left[\mathcal{W}_{2}\left( \mu_{n},\mathscr{P}_{\preceq\nu_{m}}^{\text{cx}}\right)  \right]  \right\vert \geq\frac{t}{2}\right)\\
    &\quad +\mathbb{P}\left(  \mathbb{E}\left[  \mathcal{W}_{2}\left(  \mu_{n},\mu\right)\right]  \geq\frac{t}{4}\right)  + \mathbb{P}\left(  \mathbb{E}\left[\mathcal{W}_{2}\left(  \nu_{m},\nu\right)  \right]  \geq\frac{t}{4}\right)
\end{align*}
where the last inequality follows from \Cref{thm: consistency of projected Wasserstein-1}.
Using \eqref{eq: convergence rate of FG15} and \eqref{eq: convergence rate of WB19}, if $t > 4(F(\mu, n ,d) \vee F(\nu, m, d)$ for log-Sobolev case and $t > 4(G(\mu, n ,k) \vee G(\nu, m, k)$ for bounded support case, 
\[
    \mathbb{P}\left(  \mathbb{E}\left[  \mathcal{W}_{2}\left(  \mu_{n},\mu\right)\right]  \geq\frac{t}{4}\right)=0, \quad \mathbb{P}\left(  \mathbb{E}\left[\mathcal{W}_{2}\left(  \nu_{m},\nu\right)  \right]  \geq\frac{t}{4}\right) = 0.
\]
Therefore, 
\begin{align*}
    &\mathbb{P}\left(  \left\vert \mathcal{W}_{2}\left(  \mu_{n},\mathscr{P}_{\preceq \nu_{m}}^{\text{cx}}\right)  -\mathcal{W}_{2}\left(  \mu,\mathscr{P}_{\preceq\nu}^{\text{cx}}\right)  \right\vert \geq t\right)\\
    &\leq \mathbb{P}\left(  \left\vert \mathcal{W}_{2}\left(  \mu_{n},\mathscr{P}_{\preceq\nu_{m}}^{\text{cx}}\right)  -\mathbb{E}\left[\mathcal{W}_{2}\left( \mu_{n},\mathscr{P}_{\preceq\nu_{m}}^{\text{cx}}\right)  \right]  \right\vert \geq\frac{t}{2}\right).
\end{align*}
The conclusion follow from \Cref{lem: concentration}.
\end{proof}

\section{Proofs of section~\ref{sec:computation}}\label{app: sec:computation}

\begin{proof}[Proof of \Cref{lemma : barycentric characterization}]
Let us first prove that the set in the righ-hand side of \eqref{eq : barycenter condition} is contained in $\mathscr{P}^{cx}_{\preceq \nu_m}$. Pick an arbitrary convex function $\varphi$ (that is, $\varphi \in \mathcal{A}$). For each $\omega \in \mathcal{P}(\Delta_m)$ let $\eta = (T)_{\#}(\omega)$.  It is straightforward that
\begin{align*}
    \int_{\mathbb{R}^d} \varphi(y) d\eta(y) 
    &= \int_{\Delta_m} \varphi(T(\boldsymbol{\alpha})) d\omega(\boldsymbol{\alpha})\\
    &= \int_{\Delta_m} \varphi\left(\sum_{j=1}^m \alpha_j y_j\right) d\omega(\boldsymbol{\alpha})\\
    &\leq \int_{\Delta_m} \sum_{j=1}^m \alpha_j \varphi( y_j) d\omega(\boldsymbol{\alpha})
\end{align*}
where the inequality follows by Jensen's inequality. Hence, it follows from the condition on the right in \eqref{eq : barycenter condition} that
\[
    \int_{\mathbb{R}^d} \varphi(y) d\eta(y) \leq \int_{\Delta_m} \sum_{j=1}^m \alpha_j \varphi( y_j) d\omega(\boldsymbol{\alpha}) = \sum_{j=1}^m \varphi( y_j) \int_{\Delta_m} \alpha_j  d\omega(\boldsymbol{\alpha})= \int_{\mathbb{R}^d} \varphi(y) d\nu_m(y).
\]

For the converse direction, let $\eta \in \mathscr{P}^{cx}_{\preceq \nu_m}$. Then, there exists a martingale coupling $\pi$ between $\eta$ and $\nu_m$. By the disinegration theorem,
\[
    d\pi(y, y') = dp(y' | y) d\eta(y)
\]
where $p(y' | y)$ is a probability distribution over $\spt(\nu_m)=\{y_1, \dots, y_m\}$. In other words,
\[
    p(y' | y) = \sum_{i=1}^m p(y_j |y) \delta_{y_j}.
\]
Furthermore, since $\pi$ is a martingale coupling between $\eta$ and $\nu_m$,
\[
    \int_{\mathbb{R}^d} y' d p(y' | y) = \sum_{i=1}^m y_j p(y_j |y)  = y
\]
$\eta$-almost surely, and for any $1 \leq j \leq m$,
\[
    \int_{\mathbb{R}^d} p(y_j |y) \delta_{y_j} d\eta(y) = \nu_{m}(y_j) \delta_{y_j}.
\]

Let's define $S : 
\spt(\eta)
\longrightarrow  \Delta_m$ by $S(y):= (p(y_1 |y), \dots, p(y_m |y))^T$, in $\eta$-almost everywhere sense. Let $\omega:= (S)_{\#}(\eta)$, a probability distribution over $\Delta_m$. It is straightforward that for all $1 \leq j \leq m$,
\[
    \int_{\Delta_m} \alpha_j d\omega(\boldsymbol{\alpha}) =  \int_{\mathbb{R}^d} p(y_j |y) d\eta(y) = \nu_{m}(y_j),
\]
which verifies that $\omega$ satisfies \eqref{eq : barycenter condition}. This completes the proof.
\end{proof}

\begin{proof}[Proof of \Cref{thm : computing projection of mu_n}]
First, the objective function of \eqref{eq: minimization for projection} is quadratic while its feasible set is the set of $\boldsymbol{\alpha}: \{ x_1, \cdots, x_n\} \to \R^m$ satisfying the given constraints. So it is a compact set, e.g. with respect to the uniform topology for functions. Since the feasible set is also convex, due to the strict convexity of the objective function we have a unique minimizer. 
 We will show  below that the two problems \eqref{eq: minimization for projection} and the projection problem
 $\inf_{\eta \in \mathscr{P}^{cx}_{\preceq \nu_m}} \mathcal{W}^2_2 \left( \mu_n, \eta \right)$ have the same optimal solutions (which is unique as we just saw).

Notice from Lemma \ref{thm : characterization of projected measure} that there exists some $1$-Lipschitz map $\Phi: \R^d\to \R^d$, depending on $\mu_n$, such that
\[
    \overline{\mu}_n  = (\Phi)_{\#}(\mu_n),
\]
which achieves $\inf_{\eta \in \mathscr{P}^{cx}_{\preceq \nu_m}} \mathcal{W}^2_2 \left( \mu_n, \eta \right)$. A point is that  $\overline{\mu}_n$ is supported on the discrete set $\{\Phi(x_1), \cdots, \Phi(x_n) \}$ and 
\begin{align*}
    \overline{\mu}_n = \sum_{i=1}^n \mu_n (x_i) \delta_{\Phi(x_i)}.
\end{align*}
On the other hand, 
since $\overline{\mu}_n \in \mathscr{P}^{cx}_{\preceq \nu_m}$, by \eqref{eq : barycenter condition}, there exists some $\omega \in \mathcal{P}(\Delta_m)$ such that
\[
    \overline{\mu}_n = (T)_{\#}(\omega)
\]
where $T(\boldsymbol{\alpha}) = \sum_{j=1}^m \alpha_j y_j$, satisfying
\begin{align*}
   \int_{\Delta_m} \alpha_j d \omega(\boldsymbol{\alpha}) = \sum_{i=1}^n \alpha_j(x_i) \mu_n(x_i) = \nu_m(y_j).
\end{align*}
We can choose for each $\Phi(x_i)$, the corresponding $\boldsymbol{\alpha}= \boldsymbol{\alpha} (x_i)$
with $\Phi(x_i ) = T(\boldsymbol{\alpha}(x_i)).$
Then, $\overline{\mu}_n$ is represented as 
\begin{align}\label{eqn:bar-mu-n-delta}
    \overline{\mu}_n = \sum_{i=1}^n \mu_n(x_i) \delta_{T(\boldsymbol{\alpha}(x_i))}.
\end{align}
Thus the optimization problem for finding the projection problem $\inf_{\eta \in \mathscr{P}^{cx}_{\preceq \nu_m}} \mathcal{W}^2_2 \left( \mu_n, \eta \right)$ is restricted to those of $\overline{\mu}_n$ represented as \eqref{eqn:bar-mu-n-delta} satisfying the conditions given in \eqref{eq : barycenter condition} which are nothing but the constraints of  \eqref{eq: minimization for projection}. 
In other words, the feasible set of \eqref{eq: minimization for projection} contains all the optimizers of $\inf_{\eta \in \mathscr{P}^{cx}_{\preceq \nu_m}} \mathcal{W}^2_2 \left( \mu_n, \eta \right)$. On the other hands, the feasible set of \eqref{eq: minimization for projection} is obviously contained in the feasible set of  $\inf_{\eta \in \mathscr{P}^{cx}_{\preceq \nu_m}} \mathcal{W}^2_2 \left( \mu_n, \eta \right)$, because of 
\eqref{eq : barycenter condition}. 
Now notice that the objective functions of these two problems are the same since  
$\mathcal{W}^2_2 \left( \mu_n, \overline{\mu}_n \right)$ is exactly
\[
    \mathcal{W}^2_2 \left( \mu_n, \overline{\mu}_n \right)
    =  \sum_{i=1}^n \mu_n(x_i) \left( \sum_{j=1}^m \alpha^*_j(x_i) y_j - x_i \right)^2.
\]
Therefore we conclude that the two problems  problem  \eqref{eq: minimization for projection} and  
$\inf_{\eta \in \mathscr{P}^{cx}_{\preceq \nu_m}} \mathcal{W}^2_2 \left( \mu_n, \eta \right)$ have exactly the same optimal solutions; which is unique as we already discussed above. This concludes the proof.
\end{proof}
\end{appendix}

%

\end{document}